\documentclass[11pt]{article}

\usepackage[T1]{fontenc}
\usepackage[utf8]{inputenc}

\usepackage{amsmath,amssymb,amsfonts}
\usepackage{amsthm}
\usepackage{graphicx}
\usepackage{hyperref}

\theoremstyle{plain}
\newtheorem{theorem}{Theorem}
\newtheorem{lemma}[theorem]{Lemma}
\newtheorem{proposition}[theorem]{Proposition}
\newtheorem{corollary}[theorem]{Corollary}

\theoremstyle{definition}
\newtheorem{definition}[theorem]{Definition}

\theoremstyle{remark}
\newtheorem{remark}[theorem]{Remark}

\title{Entropy of pebble and space complexity}

\author{
J. Andres Montoya\\
Universidad Nacional de Colombia\\
\texttt{jamontoyaa@unal.edu.co}
}

\date{}

\begin{document}

\maketitle

\begin{abstract}
Let $\mathcal{L}$ denote the class Logspace and $\mathcal{NL}$ the class
NLogspace. We use $\log \mathcal{CFL}$ to denote the closure under logspace
reductions of the set of context-free languages. We prove that
$\mathcal{NL} \neq \log \mathcal{CFL}$. This result implies
$\mathcal{L} \neq \mathrm{P}$ and the stronger separation
$\mathcal{NL} \neq \mathrm{P}$.
\end{abstract}


This work was created entirely by the sole author; any remaining errors are
the sole responsibility of the author. This work received no funding.
\section{Introduction}

\begin{definition}
We use $\mathcal{CFL}$ to denote the class of context-free languages. We use 
$\log \mathcal{CFL}$ to denote the closure of $\mathcal{CFL}$ under logspace
reductions.
\end{definition}

Let $\mathcal{P}$ denote the class Ptime. We have:

\begin{theorem}
$\mathcal{NL}\subseteq \log \mathcal{CFL\subseteq P}.$
\end{theorem}

\begin{proof}
There exist context-free languages that are $\mathcal{NL}$-complete; see 
\cite{Sud}. Then, $\mathcal{NL}\subseteq \log \mathcal{CFL}$.

The class $\mathcal{CFL}$ is contained in Dtime$\left( n^{3}\right) $; see 
\cite{Younger}. Then, $\log \mathcal{CFL\subseteq P}.$
\end{proof}

The class $\mathcal{NL}$ can be characterised as the union of the levels of
the nondeterministic pebble hierarchy; see Theorem \ref{Nondeterministic}.
We exploit this fact to attack the problem $\mathcal{NL}$ $vs.$ $\log 
\mathcal{CFL}.$ We show that $\mathcal{NL}$ is strictly contained in $\log 
\mathcal{CFL}$. To this end, we construct a sequence of languages contained
in $\mathcal{NL}$, denoted by $\left\{ \text{RA}_{k}\right\} _{k\geq 0}$,
and satisfying the following properties:

\begin{enumerate}
\item There exists a language RA$_{\infty }$ in $\log \mathcal{CFL}$ such
that, for all $m\geq 2$, if RA$_{\infty }$ belongs to the $m$-th level of
the nondeterministic pebble hierarchy, denoted $\mathcal{NREG}_{m},$ then
the entire sequence $\left\{ \text{RA}_{k}\right\} _{k\geq 0}$ is contained
in $\mathcal{NREG}_{m}$. We say that RA$_{\infty }$ is an \textit{upper bound%
} for $\left\{ \text{RA}_{k}\right\} _{k\geq 0}$ with respect to the
nondeterministic pebble hierarchy.

\item The sequence $\left\{ \text{RA}_{k}\right\} _{k\geq 0}$ is \textit{high%
} in the nondeterministic pebble hierarchy, meaning that, for every $m\geq 1$%
, there exists $k$ such that RA$_{k}$ does not belong to $\mathcal{NREG}_{m}$%
.
\end{enumerate}

It is easy to check that RA$_{\infty }$ cannot belong to $\mathcal{NL}$.

\begin{lemma}
Suppose that $\{\mathrm{RA}_k\}_{k \geq 0}$ is high in the
nondeterministic pebble hierarchy and that
$\mathrm{RA}_{\infty}$ is an upper bound for this sequence.
Then the language $\mathrm{RA}_{\infty}$ does not belong to
$\mathcal{NL}$.
\end{lemma}

\begin{proof}
Suppose that$\left\{ \text{RA}_{k}\right\} _{k\geq 0}$ is \textit{high} in
the nondeterministic pebble hierarchy and that RA$_{\infty }\in \mathcal{NL}$%
. Then, there exists $m$ such that RA$_{\infty }\in \mathcal{NREG}_{m}$.
This implies that the entire sequence $\left\{ \text{RA}_{k}\right\}
_{k\geq 0}$ is\ included in $\mathcal{NREG}_{m}$. We get a contradiction.
The lemma follows.
\end{proof}

We construct $L_{G_{0}}\in \mathcal{CFL}$ and show that $L_{G_{0}}$ is an
upper bound for every sequence of context-free languages; (see Section \ref%
{Section3}). This reduces our problem to that of constructing a sequence of
context-free languages that is high in the nondeterministic pebble
hierarchy. We construct such a sequence, denoted $\left\{ \text{RA}%
_{k}\right\} _{k\geq 0}$, and we show that it is high; (see Sections \ref%
{Section4} to \ref{Section9}). We employ entropy arguments to prove this.
The connection with entropy\ is based on the following naive idea: languages
that force automata accepting them to store high-entropy random variables
are languages that require large work tapes. Let us briefly discuss this
issue.

Let $\mathcal{A}$ be a nondeterministic pebble automaton accepting $%
L\subseteq \Sigma ^{\ast }$ and $k_{\mathcal{A}}$ the number of pebbles used
by $\mathcal{A}.$ Let $\mathcal{S}$ be an infinite subset of $\Sigma ^{\ast
} $. Let $n\geq 1$. We associate with the triple $\left( \mathcal{A},%
\mathcal{S},n\right) $ a random variable $X_{\mathcal{A}}\left( \mathcal{S}%
,n\right) $ that records the locations of $\mathcal{A}$'s pebbles while
processing inputs from the set $\mathcal{S}\cap \Sigma ^{n}$. For all $s\geq
0$, if the entropy of $X_{\mathcal{A}}\left( \mathcal{S},n\right) $, denoted 
$H\left( X_{\mathcal{A}}\left( \mathcal{S},n\right) \right) $, exceeds $%
s\cdot \log \left( n\right) $ for infinitely many values of $n$, it follows
that $k_{\mathcal{A}}\geq s$.

Let $\Sigma =\left\{ 0,1,\$\right\} .$ Let $k\geq 1$, and let $\mathcal{A}%
_{k}$ be a nondeterministic pebble automaton accepting the language RA$%
_{k}\subseteq \Sigma _{k}^{\ast }$. We show that there exists an infinite set 
$\mathcal{S}_{k}\subseteq \Sigma _{k}^{\ast }$ such that, for all $d>\frac{1%
}{8}$ and all $\gamma >0$, the following inequality holds asymptotically:%
\begin{equation*}
H\left( X_{\mathcal{A}_{k}}\left( \mathcal{H}_{k}^{\Sigma },n\right) \right)
\geq d\cdot \left( 1-\gamma \right) \cdot \log \left( n\right) .
\end{equation*}%
This result entails the highness of $\left\{ \text{RA}_{k}\right\} _{k\geq
1} $ and the separation $\mathcal{NL}\neq \log \mathcal{CFL}$.

\subsection{Notation and Terminology}

Let $\Sigma $ be a finite alphabet. We use $\Sigma ^{\ast }$ to denote the
set of all finite words over $\Sigma $. We use $\Sigma ^{+}$ to denote the
set $\Sigma ^{\ast }-\left\{ \xi \right\} ,$ where $\xi $ denotes the empty
word.

Let $n\geq 2,$ we use $\Sigma ^{n}$ to denote the subset of $\Sigma ^{\ast }$
constituted by the words of length $n.$

Let $w\in \Sigma ^{\ast }$. We use $\left\vert w\right\vert $ to denote the
length of $w.$ Given $i<j\leq \left\vert w\right\vert $, we use $w\left[ i%
\right] $ to denote the $i$-th letter of $w,$ and we use $w\left[ i,\ldots ,j%
\right] $ to denote the factor (infix) $w\left[ i\right] w\left[ i+1\right]
\cdots w\left[ j\right] $.

Let $L\subseteq \Sigma ^{\ast },$ we use $co$-$L$ to denote the language 
\begin{equation*}
\Sigma ^{\ast }-L=\left\{ w\in \Sigma ^{\ast }:w\notin L\right\} .
\end{equation*}

Let $n\geq 1$, we use $\log \left( n\right) $ to denote the logarithm of $n$
with respect to the basis $2$. We use $\left\lfloor \log \left( n\right)
\right\rfloor $ to denote the lower integer part of $\log \left( n\right) $.
Notice that, for all $n\geq 1$, the following inequalities hold: 
\begin{equation*}
\frac{n}{2}\leq 2^{\left\lfloor \log \left( n\right) \right\rfloor }\leq n.
\end{equation*}

Let $G$ be a digraph, we use $V\left( G\right) $ and $E\left( G\right) $ to
denote the set of nodes and the set of edges of $G$.

\section{$\mathcal{L}$, $\mathcal{NL}$, and the pebble hierarchy\label%
{Section2}}

Let $k\geq 0$. Deterministic $k$-pebble automata are deterministic two-way
finite state automata provided with $k$ labeled pebbles, which the finite
control can use as markers on the input tape. The study of this restricted
type of Turing machine goes back to Kreider and Ritchie \cite{Kreider}.

\begin{definition}
\label{PebbleAutomata}Let $k\geq 0,$ a deterministic $k$-pebble automaton is
a tuple%
\begin{equation*}
\mathcal{M}=\left( Q,q_{0},\Sigma ,A,\delta \right) ,
\end{equation*}%
where:

\begin{enumerate}
\item $Q$ is a finite set of states.

\item $q_{0}\in Q$ is the initial state.

\item $A\subseteq Q$ is the set of accepting states.

\item $\Sigma $ is the input alphabet.

\item The transition function 
\begin{equation*}
\delta :Q\times \Sigma \times \left( \mathcal{P}\left( \left\{
1,...,k\right\} \right) \right) ^{2}\rightarrow Q\times \left( \mathcal{P}%
\left( \left\{ 1,...,k\right\} \right) \right) ^{2}\times \left\{
-1,0,1\right\}
\end{equation*}%
is deterministic.
\end{enumerate}

Let $\mathcal{M}$ be a $k$-pebble automaton. Automaton $\mathcal{M}$ is a
two-way deterministic finite state automaton provided with $k$ labeled
pebbles. This automaton has the following capabilities:

\begin{itemize}
\item It can place its pebbles on the tape.

\item It can sense the pebbles that lie on the current cell.

\item It can pick specific pebbles from the set of pebbles that lie on the
current cell.
\end{itemize}

Let $\left( q,a,B,C\right) $ be a tuple such that $B$ and $C$ are disjoint
subsets of $\left\{ 1,...,k\right\} .$ Suppose that $\mathcal{M}$ is
processing the input $w.$ Suppose that:

\begin{itemize}
\item $q$ is the current inner state.

\item $a$ is the character being scanned.

\item $B$ is the set of pebbles that are placed on the current cell.

\item $C$ is the set of available pebbles.
\end{itemize}

Suppose $\delta \left( q,a,B,C\right) =\left( p,D,K,\varepsilon \right) .$
Then:

\begin{itemize}
\item $B\cup C=D\cup K,$

\item $\mathcal{M}$ change its inner state from $q$ to $p,$

\item $\mathcal{M}$ places the pebbles belonging to $D$ (and only these
pebbles) on the current cell, and

\item $\mathcal{M}$ moves its head in the direction indicated by $%
\varepsilon \in \left\{ -1,0,1\right\} $.
\end{itemize}
\end{definition}

Let $\mathcal{M}$ be a deterministic pebble automaton and let $w$ be an
input of $\mathcal{M}$. We say that $\mathcal{M}$ accepts $w$ if this input
leads the automaton to an accepting state.

\begin{definition}
We use $\mathcal{REG}_{m}$ to denote the class of languages accepted by
deterministic $m$-pebble automata.
\end{definition}

The \textit{deterministic pebble hierarchy} is the chain of inclusions:

\begin{equation*}
\mathcal{REG}_{0}\subseteq \mathcal{REG}_{1}\subseteq \mathcal{REG}%
_{2}\subseteq \cdots
\end{equation*}

We have the following result:

\begin{theorem}
\label{PebblesTheorem}
The equality
\[
\mathcal{L} = \bigcup_{0 \leq m} \mathcal{REG}_m
\]
holds.
\end{theorem}

\begin{proof}
Let $\mathcal{M}$ be an $m$-pebble automaton. Observe that a single pebble
of $\mathcal{M}$ can be simulated using a binary work tape of logarithmic
size. Hence, it is straightforward to construct a deterministic Turing
machine $\mathcal{N}$ equipped with $m+1$ work tapes of logarithmic size
that simulates $\mathcal{M}$. Therefore,
\[
\bigcup_{0 \leq m}\mathcal{REG}_m \subseteq \mathcal{L}.
\]

We now prove the converse inclusion. Let $L \in \mathcal{L}$. Then there
exist an integer $m \geq 0$ and a deterministic Turing machine
$\mathcal{M}$ such that:

\begin{itemize}
\item $\mathcal{M}$ accepts $L.$

\item $\mathcal{M}$ is provided with a read-only input tape and $m$ binary
tapes of size $\left\lfloor \log \left( n\right) \right\rfloor -2$.
\end{itemize}

Let us prove that each one of those binary tapes can be simulated using
three pebbles. Let $w\in \Sigma ^{n}$ be the input of $\mathcal{M}$. Let us
focus on the $i$-th tape and let us consider the configuration reached by
this tape at instant $t.$ We represent this configuration as a pair $\left(
1n_{L},n_{R}1\right) \in \left( \left\{ 0,1\right\} ^{\ast }\right) ^{2}.$
This pair tells us that $n_{L}n_{R}$ is the work tape content, and it also
tells us that the head assigned to this tape is located on cell $\left\vert
n_{L}\right\vert .$ Let $\left( n_{R}1\right) ^{R}$ be the reverse of $%
n_{R}1.$ Note that $1n_{L}$ and $\left( n_{R}1\right) ^{R}$ are binary
strings whose lengths are bounded above by $\left\lfloor \log \left(
n\right) \right\rfloor -1.$ These binary strings encode two positive
integers $m_{L}$ and $m_{R}$ that belong to the interval $\left\{
0,...,n-1\right\} $. We use this to represent the configuration $\left(
1n_{L},n_{R}1\right) $ using two pebbles that we call $1$ and $2$. We assign
to pebble $1$ the position $m_{L}$ and to pebble $2$ the position $m_{R}$.
We use these pebbles to simulate the changes occurring on tape $i.$ Let us
suppose, for instance, that the transition function of $\mathcal{M}$ forces
the head (of tape $i$) to move one cell to the left, after replacing with $0$
the character $1$ that was written on the current cell (which is cell $%
\left\vert n_{L}\right\vert $). Let $\left( 1n_{L}^{\ast },n_{R}^{\ast
}1\right) $ be the configuration reached at time $t+1$. Let $m_{L}^{\ast }$
and $m_{R}^{\ast }$ be the corresponding integers. Notice that%
\begin{equation*}
m_{L}^{\ast }=\frac{m_{L}-1}{2}\text{ and }m_{R}^{\ast }=2m_{R}.
\end{equation*}%
This means that we have to move pebbles $1$ and $2,$ from the positions $%
m_{L}$ and $m_{R}$ to the positions $\frac{m_{L}-1}{2}$ and $2m_{R}.$ This
can be done with the help of a third pebble. This means that we can simulate
this transition using three pebbles. Notice that we can simulate any other
transition of the $i$-th tape using the same three pebbles. This means that
we can replace tape $i$ with three pebbles. This also means that we can
replace all the $m$ work tapes of $\mathcal{M}$ with $3m$ pebbles. If we do
the latter, we obtain a deterministic $3m$-pebble automaton that accepts $L.$
We get that $L\in \mathcal{REG}_{3m}$. We conclude that $\mathcal{L}$ is
included in
\[
\bigcup_{0 \leq m}\mathcal{REG}_m .
\]
The theorem is
proved.
\end{proof}

Let $m\geq 0$. Nondeterministic $m$-pebble automata are the nondeterministic
counterparts of $m$-pebble automata.

\begin{definition}
Let $m\geq 0$, we use $\mathcal{NREG}_{m}$ to denote the class of languages
accepted by nondeterministic $m$-pebble automata.
\end{definition}

The \textit{nondeterministic pebble hierarchy} is the chain of inclusions:

\begin{equation*}
\mathcal{NREG}_{0}\subseteq \mathcal{NREG}_{1}\subseteq \mathcal{NREG}%
_{2}\subseteq \cdots
\end{equation*}%
The proof of the following theorem is completely analogous to the proof of
the previous one.
\begin{theorem}
\label{Nondeterministic}
The equality
\[
\mathcal{NL} = \bigcup_{0 \leq m}\mathcal{NREG}_m
\]
holds.
\end{theorem}

\section{Word morphism and Greibach hardest context-free language\label%
{Section3}}

We want to construct a sequence $\left\{ \text{RA}_{k}\right\} _{k\geq
0}\subseteq \mathcal{CFL}$ and a language RA$_{\infty }\in \mathcal{CFL}$
such that:

\begin{enumerate}
\item $\left\{ \text{RA}_{k}\right\} _{k\geq 0}$ is high in the
nondeterministic pebble hierarchy.

\item RA$_{\infty }$ is an upper bound for $\left\{ \text{RA}_{k}\right\}
_{k\geq 0}$.
\end{enumerate}

We show that there exists $L_{G_{0}}\in \mathcal{CFL}$ that is an upper
bound for every sequence $\left\{ L_{k}\right\} _{k\geq 0}\subseteq \mathcal{%
CFL}$. This reduces our problem to that of constructing a \textit{high
sequence} of context-free languages.

\begin{definition}
Let $\Sigma $ and $\Omega $ be two finite alphabets. Let $f$ be a function
from $\Sigma $ to $\Omega ^{\ast }$. Let $\widehat{f}:\Sigma ^{\ast }$ $%
\rightarrow $ $\Omega ^{\ast }$ be the function that is defined by the
equation%
\begin{equation*}
\widehat{f}\left( w\right) =\left\{ 
\begin{array}{c}
\varepsilon \text{, if }w=\varepsilon \\ 
f\left( w\left[ 1\right] \right) \cdots f\left( w\left[ \left\vert
w\right\vert \right] \right) \text{, if }\left\vert w\right\vert >0%
\end{array}%
\right.
\end{equation*}%
We say that $\widehat{f}$ is a word morphism from $\Sigma ^{\ast }$\ to $%
\Omega ^{\ast }$, and any word morphism between these two free monoids is
defined this way from a function $f:\Sigma $ $\rightarrow $ $\Omega ^{\ast }$%
.
\end{definition}

\begin{definition}
Let $L\subset \Sigma ^{\ast }$ and $T\subset \Gamma ^{\ast }$be two
languages. We say that $L$ is an inverse image of $T$ if and only if there
exists a word morphism $\widehat{f}:\Sigma ^{\ast }\rightarrow $ $\Omega
^{\ast }$ such that $L=\widehat{f}^{-1}\left( T\right) $.
\end{definition}

The class $\mathcal{CFL}$ contains languages that are complete for $\mathcal{%
CFL}$ under inverse images of \textit{word morphisms}; see \cite{Greibach}.
This means that there exists $L_{G_{0}}\in \mathcal{CFL}$ such that, for all 
$T\in \mathcal{CFL}$, the language $T$ is an inverse image of $L_{G_{0}}$
under a word morphism. We use $L_{G_{0}}$ to denote \textit{Greibach hardest
context-free language}; see \cite{Greibach}.

We have:

\begin{theorem}
Let $k\geq 0,$ the set $\mathcal{NREG}_{k}$ is closed under inverse images
of word morphism.
\end{theorem}

\begin{proof}
Suppose $L\subset \Sigma ^{\ast }$ is an inverse image of $T\subset \Gamma
^{\ast }$ and suppose $T\in \mathcal{NREG}_{k}$. We prove that $L\in 
\mathcal{NREG}_{k}$.

Let $\mathcal{M}$ be a nondeterministic $k$-pebble automaton accepting $%
T\subset \Gamma ^{\ast }.$ Let $f:\Sigma \rightarrow \Gamma ^{\ast }$ be a
function such that $L=\widehat{f}^{-1}\left( T\right) .$ We construct a
nondeterministic $k$-pebble automaton $\mathcal{S}$ that receives as input $%
w\in \Sigma ^{n}$ and simulates the computation of $\mathcal{M}$, on input $%
\widehat{f}\left( w\right) .$

Let 
\begin{equation*}
N_{0}=\max \left\{ \left\vert f\left( a\right) \right\vert :\text{ }a\in
\Sigma \right\}
\end{equation*}%
Let $Q$ be the set of states of $\mathcal{M}.$ The set of states of $%
\mathcal{S}$ is equal to 
\begin{equation*}
Q\times \left\{ 0,\ldots ,N_{0}\right\} ^{k+1}.
\end{equation*}%
We use this set of states to simulate $\mathcal{M}$. Recall that 
\begin{equation*}
\widehat{f}\left( w\right) =f\left( w\left[ 1\right] \right) \cdots f\left( w%
\left[ n\right] \right)
\end{equation*}%
Thus, suppose:

\begin{itemize}
\item $\mathcal{M}$ is in state $q$,

\item the input head is placed on the $r$-th character of the $i$-th factor
(i.e., factor $f\left( w\left[ i\right] \right) $),

\item the set of available pebbles is $B\subseteq \left\{ 1,\ldots
,k\right\} ,$ and

\item given $j\notin B,$ the pebble $j$ is placed on the $r_{j}$ character
of the $l_{j}$ factor.
\end{itemize}

Then:

\begin{itemize}
\item $\mathcal{S}$ is in state $\left( q,r,s_{1},\ldots ,s_{k}\right) $,
with: 
\begin{equation*}
s_{i}=\left\{ 
\begin{array}{c}
r_{i}\text{, }i\notin B \\ 
0\text{, }i\in B%
\end{array}%
\right. ,
\end{equation*}

\item the input head is placed on the $i$-th cell,

\item the set of available pebbles is $B$, and

\item for all $j\notin B$, the pebble $j$ is placed on the $l_{j}$-th cell.
\end{itemize}

The set of accepting states is the set: 
\begin{equation*}
\left\{ \left( q,r,s_{1},\ldots ,s_{k}\right) :q\in A\right\} .
\end{equation*}%
The transition function is defined accordingly.

It is easy to check that the above construction works. We conclude that $%
\mathcal{NREG}_{m}$ is closed under inverse images of word morphism. The
theorem follows.
\end{proof}

We obtain the following corollary:

\begin{corollary}
Let $m\geq 0$. If $L_{G_{0}}\in \mathcal{NREG}_{m}$ the entire class $%
\mathcal{CFL}$ is included in $\mathcal{NREG}_{m}$.
\end{corollary}

\begin{proof}
Suppose $L_{G_{0}}\in \mathcal{NREG}_{m}$. Let $L\in \mathcal{CFL}$. There
exists a word morphism $f$ such that $L=f^{-1}\left( L_{G_{0}}\right) $. The
set $\mathcal{NREG}_{m}$ is closed under inverse images of word morphism.
Then, the language $L$ belongs to $\mathcal{NREG}_{m}$ and the entire class $%
\mathcal{CFL}$ is included in $\mathcal{NREG}_{m}$.
\end{proof}

Moreover, we have:

\begin{corollary}
\label{high}There exists a sequence of context-free languages that is high
in the nondeterministic pebble hierarchy if and only if the separation $%
\mathcal{NL}\subset \log \mathcal{CFL}$ holds.
\end{corollary}

\begin{proof}
Suppose $\mathcal{NL}\subset \log \mathcal{CFL}$. Let $\left\{ L_{i}\right\}
_{i\geq 0}$ be an enumeration of $\mathcal{CFL}$. This sequence is high in
the nondeterministic pebble hierarchy.

The language $L_{G_{0}}$ is an upper bound for every sequence of languages
contained in $\mathcal{CFL}$. Then, if there exists a sequence $\left\{
L_{i}\right\} _{i\geq 0}\subseteq \mathcal{CFL}$ that is high in the
nondeterministic pebble hierarchy the language $L_{G_{0}}$ belongs to $\log 
\mathcal{CFL-NL}$.
\end{proof}

\section{High sequences of context-free languages\label{Section4}}

We showed that $\mathcal{NL}\neq \log \mathcal{CFL}$ is equivalent to the
existence of \textit{high sequences} in $\mathcal{CFL}$. We try an inductive
construction of those sequences.

\begin{definition}
Let $\Sigma ,\Gamma $ be two finite alphabets and let $\#\notin \Sigma \cup
\Gamma $. A bi-string over\ $\left( \Sigma ,\Gamma ,\#\right) $ is a string 
\begin{equation*}
w_{1}\#r_{1}\#\cdots \#w_{n}\#r_{n}
\end{equation*}%
such that $w_{1},\ldots ,w_{n}\in \Sigma ^{\ast }$; and $r_{1},\ldots
,r_{n}\in \Gamma ^{\ast }$. We use $s\left( w,r\right) $ as a shorthand for
the bi-string%
\begin{equation*}
w_{1}\#r_{1}\#\cdots \#w_{n}\#r_{n}.
\end{equation*}%
We say that $w_{1},\ldots ,w_{n}$ are the \textit{local factors} of $s\left(
w,r\right) $ and we say that $r_{1},\ldots ,r_{n}$ are the local co-factors
of this bi-string.
\end{definition}

\begin{definition}
Let $T\subset \Sigma ^{\ast }$ and let $w\in \Sigma ^{\ast }$. Define: 
\begin{equation*}
\varepsilon ^{T}\left( w\right) =\left\{ 
\begin{array}{c}
1\text{, if }w\in T \\ 
0\text{, otherwise}%
\end{array}%
\right.
\end{equation*}%
We say that $\varepsilon ^{T}\left( w\right) $ is the characteristic bit
(with respect to $T$) of the string $w$.
\end{definition}

\begin{definition}
Let $r=r_{1}\cdots r_{n}$. Let $I=\left\{ i_{1},\ldots ,i_{k}\right\} $ be a
subset of $\left\{ 1,\ldots ,n\right\} $. Suppose $i_{1}<\cdots <i_{k}$. We
use $\bigodot_{i\in I} r_i$ to denote the substring
$r_{i_1}\cdots r_{i_k}$.
\end{definition}

\begin{definition}
Let $L \subseteq \Sigma^\ast$, $T \subseteq \Gamma^\ast$, and
$\# \notin \Sigma \cup \Gamma$. We define the language
$T \otimes_{\pi} L$ as the set of all bi-strings
\[
w_1 \# r_1 \# \cdots \# w_n \# r_n
\]
that satisfy the condition
\[
\bigodot_{i \in T(w)} r_i \in L,
\]
where:
\begin{equation*}
T\left( w\right) =\left\{ i\leq n:\varepsilon ^{T}\left( w_{i}\right)
=1\right\} .
\end{equation*}%
We say that $T\otimes _{\pi }L$ is the masked concatenation of $L$ over $T$.
\end{definition}

\begin{theorem}
\label{Context-free}Let $L,T\in \mathcal{CFL}$. The language $T\otimes _{\pi
}L$ belongs to $\mathcal{CFL}$.
\end{theorem}

\begin{proof}
Let $\mathcal{A}$ and $\mathcal{B}$ be pushdown automata accepting $L$ and $%
T.$ Let $\Theta $ and $\Delta $ be the work alphabets of these automata. We
suppose that those alphabets are disjoint. We construct a pushdown automaton 
$\mathcal{N}$ accepting $T\otimes L$. The work alphabet of $\mathcal{N}$ is $%
\Theta \cup \Delta .$

Let%
\begin{equation*}
s\left( w,r\right) =w_{1}\#r_{1}\#\cdots \#v_{n}\#r_{n}
\end{equation*}%
be an input of $\mathcal{N}$. The automaton works on that input as follows:

Let $i\leq n$. Suppose that $\mathcal{N}$ is starting to process the factor $%
w_{i}.$ Automaton $\mathcal{N}$ simulates the computation of $\mathcal{B}$
on $w_{i}$. We assume:

\begin{itemize}
\item At the start of this simulation, the pushdown stack contains a string
$X_i \in \Theta^\ast$. This string $X_i$ coincides with the
content of the pushdown store of $\mathcal{A}$ after a run of
this automaton on
\[
\bigodot_{j \in T(w,i)} r_j ,
\]
where
\[
T(w,i)=\{\, j<i : \varepsilon^T(w_j)=1 \,\}.
\]

\item During the simulation of $\mathcal{B}$ in $w_i$, the automaton
$\mathcal{N}$ cannot access $X_i$.
When the simulation of $\mathcal{B}$ on $w_{i}$ ends, the automaton $%
\mathcal{N}$ knows $\varepsilon ^{T}\left( w_{i}\right) $. Then it proceeds
as follows:

\begin{enumerate}
\item $\mathcal{N}$ pops from the stack all symbols above $X_{i}$ belonging
to $\Gamma .$

\item If $w_{i}\in T$, the automaton $\mathcal{N}$ continues with the
simulation of $\mathcal{A}$ while scanning the co-factor $r_{i}$. During
this phase, $\mathcal{N}\ $is allowed to access $X_{i}$ (i.e., $\mathcal{N}\ 
$is allowed to access all the content of its pushdown store).

\item If $w_{i}\notin T$, the automaton $\mathcal{N}$ is not allowed to
continue with the simulation of $\mathcal{A}$ while scanning the co-factor $%
r_{i}$ ($\mathcal{N}$ undergoes a \textit{resting phase} while 
scanning $r_{i}$).
\end{enumerate}

\end{itemize}

$\mathcal{N}$ accepts $s\left( w,r\right) $ if and only if the simulation
ends in an accepting state of $\mathcal{A}$. 

It can be verified that $\mathcal{N}$ accepts the language $T\otimes _{\pi }L
$. The theorem follows. 
\end{proof}

\begin{definition}
Let $L$ be a context-free language. We now define a sequence $\left\{
L_{k}\right\} _{k\geq 0}\subset \mathcal{CFL}$. We define this sequence by
induction on $k$:

\begin{enumerate}
\item $L_{0}=1\cdot \left\{ 0,1\right\} ^{\ast }.$

\item Let $\Sigma _{k}$ be the alphabet of $L_{k}$, let $\#_{k+1}\notin
\Sigma _{k}$. We define
\begin{equation*}
L_{k+1}=L_{k}\otimes _{\#_{k+1}}L.
\end{equation*}
\end{enumerate}
\end{definition}

The next proposition is straightforward.

\begin{proposition}
$\mathcal{NL}\neq \log \mathcal{CFL}$ if and only if there exists $L\in 
\mathcal{CFL}$ such that $\left\{ L_{k}\right\} _{k\geq 0}$ is high in the
nondeterministic pebble hierarchy.
\end{proposition}

\begin{proof}
Let $L$ be a language in $\log \mathcal{CFL-NL}$. This language is logspace
reducible to $L_{1}$. Then $L_{1}\in \log \mathcal{CFL-NL}$ and $\left\{
L_{k}\right\} _{k\geq 0}$ is high in the nondeterministic pebble hierarchy.
On the other hand, if $\left\{ L_{k}\right\} _{k\geq 0}$ is high in the
pebble hierarchy the language $L_{G_{0}}$ belongs to $\log \mathcal{CFL-NL}$.
\end{proof}

\section{Entropy of pebble automata\label{Section5}}

\begin{definition}
Let $X_n$ be a random variable distributed over the set
$\{1,\ldots,n\}$. The Shannon entropy of $X_n$ is defined as
\[
H(X_n)
=
-\sum_{i \leq n}
\Pr(X_n=i)\,
\log\bigl(\Pr(X_n=i)\bigr),
\]
using the convention
\[
0 \cdot \log(0)=0.
\]
\end{definition}

A standard fact is that entropy is maximised by the uniform distribution:

\begin{itemize}
\item $H\left( X_{n}\right) \leq \log \left( n\right) $, and

\item the equality holds when $X_{n}$ is uniformly distributed.
\end{itemize}

\begin{definition}
Let $X$ and $Y$ be random variables. We denote by $H\left( X,Y\right) $ the
entropy of the joint random variable $\left( X,Y\right) $; this is called
the joint entropy of $X$ and $Y$.
\end{definition}

\begin{definition}
For random variables $X$ and $Y$ the conditional entropy $H\left( X\mid
Y\right) $ is defined by 
\begin{equation*}
H\left( X\mid Y\right) =H\left( X,Y\right) -H\left( Y\right) .
\end{equation*}
\end{definition}

See \cite{Shannon} for further foundational material on entropy and
information theory.

\begin{definition}
Let $\mathcal{A}$ be a nondeterministic $m$-pebble automaton accepting $%
L\subseteq \Sigma ^{\ast }$. Let $Q$ denote the set of internal states and
let $w\in \Sigma ^{n}$. The \textit{configurations} reached by $\mathcal{A}$
on input $w$ are $\left( m+2\right) $-tuples in 
\begin{equation*}
\left\{ 0,1,\ldots ,n\right\} ^{m+1}\times Q.
\end{equation*}%
Let 
\begin{equation*}
\mathcal{D}_{t}=\left( h,p_{1},\ldots ,p_{m},q\right)
\end{equation*}%
be one such tuple, representing the configuration reached by $\mathcal{A}$
at time $t$. We have:

\begin{enumerate}
\item $h$ is the current head position.

\item $p_{1},\ldots ,p_{m-1},$ and $p_{m}$ denote the locations of the $m$
pebbles (if pebble $i$ is not on the tape we set $p_{i}=0$).

\item $q\in Q$ is the internal state at time $t$.
\end{enumerate}
\end{definition}

\begin{definition}
Let $\mathcal{D}_{t}$ be as above. We call $\mathcal{D}_{t}$ the
instantaneous configuration of $\mathcal{A}$ at time $t$. Let 
\begin{equation*}
\mathcal{I}_{t}=\left( p_{1},\ldots ,p_{m},q\right) .
\end{equation*}%
We refer to $\mathcal{I}_{t}$ as the coding configuration of $\mathcal{A}$
at time $t$.
\end{definition}

\begin{definition}
\label{RandomVariable}Let $L\subset \Sigma ^{\ast }$ and let $\mathcal{A}$
be a nondeterministic $m$-pebble automaton accepting $L$. Let $\mathcal{H}$
be an infinite subset of $\Sigma ^{\ast }$ and $n\geq 1$. We define $X_{%
\mathcal{A}}\left( \mathcal{H},n\right) $ as a random variable uniformly
distributed over the set of coding configurations visited by $\mathcal{A}$
during its computations on the elements of $\mathcal{H}\cap \Sigma ^{n}$.
\end{definition}

The following theorem is straightforward.

\begin{theorem}
Let $\mathcal{A}$ be a $k$-pebble automaton accepting a language $L\subseteq
\Sigma ^{\ast }$. The following inequality holds asymptotically: 
\begin{equation*}
H\left( X_{\mathcal{A}}\left( \Sigma ^{\ast },n\right) \right) \leq \left(
k+1\right) \cdot \log \left( n\right) .
\end{equation*}
\end{theorem}

\begin{proof}
Let $\mathcal{A}$ be a $k$-pebble automaton accepting $L\subseteq \Sigma
^{\ast }$. The random variable $X_{\mathcal{A}}\left( \Sigma ^{\ast },n\right) $
is distributed over a subset of $\left\{ 1,...,n\right\} ^{k}\times Q$. Then 
\begin{equation*}
H\left( X_{\mathcal{A}}\left( \Sigma ^{\ast },n\right) \right) \leq k\cdot
\log \left( n\right) +\log \left( \left\vert Q\right\vert \right) .
\end{equation*}%
There exists $N$ such that, for all $n\geq N,$ the following inequality
holds:%
\begin{equation*}
k\cdot \log \left( n\right) +\log \left( \left\vert Q\right\vert \right)
\leq \left( k+1\right) \cdot \log \left( n\right) .
\end{equation*}%
The theorem follows.
\end{proof}

\section{Information theoretic decomposition\label{Section6}}

Let $L\subseteq \Sigma ^{\ast }$ be a language in $\mathcal{CFL}$. We study
the sequence $\left\{ L_{k}\right\} _{k\geq 0}$. The language $L_{k+1}$ is
defined by two different constraints:

\begin{itemize}
\item the local constraint encoded by $L_{k}$ and which interacts with the
local factors $w_{1},...,w_{n}$; and

\item the global constraint encoded by $L$ that interacts with the substring 
$\bigodot_{i \in L_k(w)} r_i$
\end{itemize}

We construct a set $\mathcal{H}_{k+1}^{\Sigma }$ for which these two constraints
become \textit{orthogonal} to each other. Over the constructed sets, the
entropy of pebble automata accepting $L_{k+1}$ can be decomposed as the sum
of two entropies, namely: the entropy of pebble automata accepting $L_{k}$
and the \textit{cross-entropy} of pebble automata accepting $L_{1}$, (see
below). A key element in the construction of the $\mathcal{H}_{k}^{\Sigma }$%
's is that their constituent bi-strings are made-up of a small (logarithmic)
number of factors.

\begin{remark}
\textit{An important remark about the importance of being constituted by a
polylogarithmic number of factors.} Let $\left\{ X_{n}\right\} _{n\geq 1}$
be a sequence of random variables and let $C,D,r>0.$ We have that, for all $%
\gamma >0,$ the inequality 
\begin{equation*}
H\left( X_{n}\right) \geq \left( 1-\gamma \right) \cdot C\cdot \log \left(
n\right)
\end{equation*}%
holds asymptotically if and only if, for all $\gamma >0$, the following
inequality holds asymptotically:%
\begin{equation*}
H\left( X_{n}\right) \geq \left( 1-\gamma \right) \cdot C\cdot \log \left(
D\cdot \log ^{r}\left( n\right) \cdot n\right) .
\end{equation*}
\end{remark}

\begin{definition}
Let $\Sigma ,$ $\Gamma $, and $\#$ be as above. Let $\mathcal{H}$ be an
infinite subset of $\Gamma ^{\ast }$. Suppose $\mathcal{H}\cap \Gamma
^{n}\neq \emptyset $. We use $\Sigma \left[ \mathcal{H}\right] \left( 3\cdot
\left\lfloor \log \left( n\right) \right\rfloor \cdot \left( n+1\right)
\right) $ to denote the set of all bi-strings $s\left( w,r\right) $, over $%
\left( \Sigma ,\Gamma ,\#\right) $, that satisfy the following four
conditions:

\begin{enumerate}
\item $s\left( w,r\right) $ is equal to: 
\begin{equation*}
w_{1}\#r_{1}\#\cdots \#w_{\left\lfloor \log \left( n\right) \right\rfloor
}\#r_{\left\lfloor \log \left( n\right) \right\rfloor }\#^{s}.
\end{equation*}

\item For all $1\leq i\leq \left\lfloor \log \left( n\right) \right\rfloor $%
, the string $w_{i}$ belongs to $\mathcal{H}\cap \Gamma ^{n}.$

\item For all $0\leq i\leq \left\lfloor \log \left( n\right) \right\rfloor $%
, the string $r_{i}$ belongs to $\Sigma ^{\leq n}$.

\item $s=3\cdot \left\lfloor \log \left( n\right) \right\rfloor \cdot \left(
n+1\right) -\left\vert w_{1}\#r_{1}\#\cdots \#w_{\left\lfloor \log \left(
n\right) \right\rfloor }\#r_{\left\lfloor \log \left( n\right) \right\rfloor
}\right\vert .$
\end{enumerate}
\end{definition}

\begin{definition}
We use $\Sigma[\mathcal{H}]$ to denote the set
\[
\bigcup_{n \geq 4}
\Sigma[\mathcal{H}]
\bigl(
3 \lfloor \log(n) \rfloor (n+1)
\bigr).
\]
\end{definition}

\begin{definition}
Let $\Sigma $ be a finite alphabet (the alphabet of $L$). Define:
\end{definition}

\begin{enumerate}
\item $\mathcal{H}_{0}=\left\{ 0,1\right\} ^{+}$.

\item $\mathcal{H}_{k+1}^{\Sigma }=\Sigma \left[ \mathcal{H}_{k}^{\Sigma }%
\right] .$
\end{enumerate}

\begin{remark}
Let $L\subseteq \Sigma ^{\ast }$. Consider the language $L_{1}=L_{0}\otimes
L $ and the set $\mathcal{H}_{1}^{\Sigma }\left( 3\cdot \left\lfloor \log
\left( n\right) \right\rfloor \cdot \left( n+1\right) \right) $.

Let%
\begin{equation*}
s\left( w,r\right) =w_{1}\#_{1}r_{1}\#_{1}\cdots \#_{1}w_{\left\lfloor \log
\left( n\right) \right\rfloor }\#_{1}r_{\left\lfloor \log \left( n\right)
\right\rfloor }\#_{1}^{s}
\end{equation*}%
be a bi-string in $\mathcal{H}_{1}^{\Sigma }\left( 3\cdot \left\lfloor \log
\left( n\right) \right\rfloor \cdot \left( n+1\right) \right) $. This bi-string belongs to $L_1$ if and only if
\[
\bigodot_{i : w_i[1]=1} r_i \in L .
\]So, whether $s\left( w,r\right) $ belongs to $L_{1}$
or not is completely determined by the co-factors $r_{1},...,r_{\left\lfloor
\log \left( n\right) \right\rfloor }\in \Sigma ^{\ast }$; and the short
binary string $w_{1}\left[ 1\right] \cdots w_{\left\lfloor \log \left(
n\right) \right\rfloor }\left[ 1\right] .$
\end{remark}

\begin{definition}
Let $\mathcal{A}_{1}$ be a nondeterministic pebble automaton. We say that $%
\mathcal{A}_{1}$ separates $L_{1}\cap \mathcal{H}_{1}^{\Sigma }$ from $co$-$%
L_{1}\cap \mathcal{H}_{1}^{\Sigma }$ if $\mathcal{A}_{1}$ accepts all the
strings in $L_{1}\cap \mathcal{H}_{1}^{\Sigma }$ and rejects all the strings
in $co$-$L_{1}\cap \mathcal{H}_{1}^{\Sigma }$.
\end{definition}

\begin{definition}
Let $\mathcal{A}_{1}$ be a nondeterministic pebble automaton separating $%
L_{1}\cap \mathcal{H}_{1}^{\Sigma }$ from $co$-$L_{1}\cap \mathcal{H}%
_{1}^{\Sigma }$. We use $Y_{\mathcal{A}_{1}}\left( \mathcal{H}_{1}^{\Sigma
},n\right) $ to denote a random variable that is uniformly distributed over
the set of coding configurations that are visited by $\mathcal{A}_{1}$, when
this automaton is processing an input $s\left( w,r\right) \in \mathcal{H}%
_{1}^{\Sigma }\left( 3\cdot \left\lfloor \log \left( n\right) \right\rfloor
\cdot \left( n+1\right) \right) $ and it is querying a local bit $w_{i}\left[
1\right] $ .
\end{definition}

We have:

\begin{theorem}
\label{NewProof}Let $L\subseteq \Sigma ^{\ast }$ be a context-free language.
Let $k\geq 1$ and let $\mathcal{A}_{k+1}$ be a nondeterministic pebble
automaton accepting $L_{k+1}$. There exist nondeterministic pebble automata $%
\mathcal{A}_{k}$ and $\mathcal{A}_{1}$ such that:

\begin{enumerate}
\item $\mathcal{A}_{k}$ accepts the language $L_{k}$.

\item $\mathcal{A}_{1}$ separates $L_{1}\cap \mathcal{H}_{1}^{\Sigma }$ from 
$co$-$L_{1}\cap \mathcal{H}_{1}^{\Sigma }$.

\item The following equality holds:%
\begin{eqnarray*}
&&H\left( X_{\mathcal{A}_{k+1}}\left( \mathcal{H}_{k+1}^{\Sigma },3\cdot
\left\lfloor \log \left( n\right) \right\rfloor \cdot \left( n+1\right)
\right) \right) \\
&=&H\left( X_{\mathcal{A}_{k}}\left( \mathcal{H}_{k}^{\Sigma },n\right)
\right) +H\left( Y_{\mathcal{A}_{1}}\left( \mathcal{H}_{1}^{\Sigma
},n\right) \right) .
\end{eqnarray*}
\end{enumerate}
\end{theorem}

\begin{proof}
Let $\mathcal{A}_{k+1}$ be a nondeterministic pebble automaton accepting $%
L_{k+1}$. Let $n\geq 2$ and suppose that $\mathcal{H}_{k}^{\Sigma }\left(
n\right) $ is non-empty. We analyze the entropy of the random variable 
\begin{equation*}
X_{\mathcal{A}_{k+1}}\left( \mathcal{H}_{k+1}^{\Sigma },3\cdot \left\lfloor
\log \left( n\right) \right\rfloor \cdot \left( n+1\right) \right) .
\end{equation*}%
Let $\mathcal{C}_{k+1}$ be an outcome of this variable. We assume that $%
\mathcal{A}_{k+1}$ is checking a local factor $w_{i}\in \mathcal{H}%
_{k}^{\Sigma }\left( n\right) $ and computing the characteristic bit $%
\varepsilon ^{L_{k}}\left( w_{i}\right) $. We use $\mathcal{A}_{k}$ to
denote the nondeterministic pebble automaton that accepts $L_{k}$ and which
is used by $\mathcal{A}_{k+1}$ in order to compute the $\varepsilon ^{L_{k}}$%
-bits. The configuration $\mathcal{C}_{k+1}$ encodes a \textit{local
configuration} $\mathcal{C}_{k}$: the current configuration accessed by $%
\mathcal{A}_{k}$ on the local factor $w_{i}\in \mathcal{H}_{k}^{\Sigma
}\left( n\right) $. The following equality holds:%
\begin{eqnarray*}
&&H\left( X_{\mathcal{A}_{k+1}}\left( \mathcal{H}_{k+1}^{\Sigma },3\cdot
\left\lfloor \log \left( n\right) \right\rfloor \cdot \left( n+1\right)
\right) \right) \\
&=&H\left( \mathcal{C}_{k}\right) +H\left( X_{\mathcal{A}_{k+1}}\left( 
\mathcal{H}_{k+1}^{\Sigma },3\cdot \left\lfloor \log \left( n\right)
\right\rfloor \cdot \left( n+1\right) \right) \mid \mathcal{C}_{k}\right) .
\end{eqnarray*}

The random variable $\mathcal{C}_{k}$ is uniformly distributed over the set
of coding configurations visited by $\mathcal{A}_{k}$ on inputs from $%
\mathcal{H}_{k}^{\Sigma }\left( n\right) $. This means that $\mathcal{C}_{k}$
is equivalent to $X_{\mathcal{A}_{k}}\left( \mathcal{H}_{k}^{\Sigma
},n\right) $ and the following equality holds:%
\begin{equation*}
H\left( \mathcal{C}_{k}\right) =H\left( X_{\mathcal{A}_{k}}\left( \mathcal{H}%
_{k}^{\Sigma },n\right) \right) .
\end{equation*}

We observe that:

\begin{itemize}
\item $\mathcal{A}_{k+1}$ uses the nondeterministic local subroutine $%
\mathcal{A}_{k}$ to compute the $\varepsilon ^{L_{k}}$-bits, and

\item it uses a global subroutine $\mathcal{B}_{k+1}$ for computing the next
factor to be processed, processing the co-factors, and deciding whether the
bi-string $s\left( w,r\right) $ is accepted or rejected.
\end{itemize}

Let us show that we can use $\mathcal{B}_{k+1}$ for separating $L_{1}\cap 
\mathcal{H}_{1}^{\Sigma }$ from $co$-$L_{1}\cap \mathcal{H}_{1}^{\Sigma }$.
Thus, let 
\begin{equation*}
s\left( w,r\right) =w_{1}\#_{1}r_{1}\#_{1}\cdots \#_{1}w_{\left\lfloor \log
\left( n\right) \right\rfloor }\#_{1}r_{\left\lfloor \log \left( n\right)
\right\rfloor }\#_{1}^{s}\in \mathcal{H}_{1}^{\Sigma }.
\end{equation*}%
For all $i\leq n$, we have: $\varepsilon ^{L_{0}}\left( w_{i}\right) =w_{i}%
\left[ 1\right] $. We work on this instance as follows:

\begin{enumerate}
\item We simulate, on $s\left( w,r\right) $, the computation of $\mathcal{A}%
_{k+1}$ on the string%
\begin{equation*}
w_{1}\#_{k+1}r_{1}\#_{k+1}\cdots \#_{k+1}w_{\left\lfloor \log \left(
n\right) \right\rfloor }\#_{k+1}r_{\left\lfloor \log \left( n\right)
\right\rfloor }\#_{k+1}^{s}.
\end{equation*}

\item Each time $\mathcal{A}_{k+1}$ calls the local subroutine $\mathcal{A}%
_{k}$ we proceed as follows: we interrupt the execution of this subroutine
(on the current local factor $\#_{1}w_{i}$), ask for the value of $w_{i}%
\left[ 1\right] $, and set 
\begin{equation*}
\varepsilon ^{L_{k}}\left( w_{i}\right) =w_{i}\left[ 1\right] .
\end{equation*}
\end{enumerate}

We use $\mathcal{B}_{1}$ to denote this truncated version of $\mathcal{B}%
_{k+1}$. Note that $\mathcal{B}_{1}$ is a nondeterministic pebble automaton
that separates $L_{1}\cap \mathcal{H}_{1}^{\Sigma }$ from $co$-$L_{1}\cap 
\mathcal{H}_{1}^{\Sigma }$. Note also that%
\begin{equation*}
H\left( Y_{\mathcal{B}_{1}}\left( \mathcal{H}_{1}^{\Sigma },n\right) \right)
=H\left( X_{\mathcal{A}_{k+1}}\left( \mathcal{H}_{k+1}^{\Sigma },3\cdot
\left\lfloor \log \left( n\right) \right\rfloor \cdot \left( n+1\right)
\right) \mid \mathcal{C}_{k}\right) .
\end{equation*}

The theorem follows.
\end{proof}

\begin{definition}
Let $L\subseteq \Sigma ^{\ast }$ be a language in $\mathcal{CFL\cap L}$. We
say that $L$ is hard-to-cross if there exists $d>0$ such that, for all
nondeterministic pebble automata $\mathcal{A}_{1}$ separating $L_{1}\cap 
\mathcal{H}_{1}^{\Sigma }$ from $co$-$L_{1}\cap \mathcal{H}_{1}^{\Sigma }$
and all $\gamma >0,$ the following inequality holds asymptotically:%
\begin{equation*}
H\left( Y_{\mathcal{A}_{1}}\left( \mathcal{H}_{1}^{\Sigma },n\right) \right)
\geq d\cdot \left( 1-\gamma \right) \cdot \log \left( n\right) .
\end{equation*}
\end{definition}

We obtain the following corollary:

\begin{corollary}
\label{corolarioInductivo}Let $L\in \mathcal{CFL\cap NL}$ be a hard-to-cross
language, the sequence $\left\{ L_{k}\right\} _{k\geq 0}$ is high in the
nondeterministic pebble hierarchy.
\end{corollary}

\begin{proof}
We show that, for all $\gamma \geq 0$ and all $\mathcal{A}_{k},$ the
following inequality holds asymptotically:%
\begin{equation*}
H\left( X_{\mathcal{A}_{k}}\left( \mathcal{H}_{k}^{\Sigma },n\right) \right)
\geq k\cdot d\cdot \left( 1-\gamma \right) \cdot \log \left( n\right) .
\end{equation*}

We use induction on $k$.

The non-negativity of Shannon entropy implies that the inequality holds for $%
k=0$. 

We assume that the inequality holds for $k$ and show that it holds for $k+1.$
Thus, let $\mathcal{A}_{k+1}$ be a nondeterministic pebble automaton
accepting $L_{k+1}$. Theorem \ref{NewProof} asserts that there exist $%
\mathcal{A}_{k}$ and $\mathcal{A}_{1}$ such that%
\begin{eqnarray*}
&&H\left( X_{\mathcal{A}_{k+1}}\left( \mathcal{H}_{k+1}^{\Sigma },3\cdot
\left\lfloor \log \left( n\right) \right\rfloor \cdot \left( n+1\right)
\right) \right)  \\
&=&H\left( X_{\mathcal{A}_{k}}\left( \mathcal{H}_{k}^{\Sigma },n\right)
\right) +H\left( Y_{\mathcal{A}_{1}}\left( \mathcal{H}_{1}^{\Sigma
},n\right) \right) .
\end{eqnarray*}%
The inductive hypothesis ensures that, for all $\gamma >0$, the following
inequality holds asymptotically:%
\begin{equation*}
H\left( X_{\mathcal{A}_{k}}\left( \mathcal{H}_{k}^{\Sigma },n\right) \right)
\geq k\cdot d\cdot \left( 1-\gamma \right) \cdot \log \left( n\right) 
\end{equation*}%
On the other hand, since $L$ is hard-to-cross, for all $\gamma >0$, the
following inequality holds asymptotically:%
\begin{equation*}
H\left( Y_{\mathcal{A}_{1}}\left( \mathcal{H}_{1}^{\Sigma },n\right) \right)
\geq d\cdot \left( 1-\gamma \right) \cdot \log \left( n\right) .
\end{equation*}%
Then, for all $\gamma >0$, the inequality%
\begin{equation*}
H\left( X_{\mathcal{A}_{k+1}}\left( \mathcal{H}_{k+1}^{\Sigma },n\right)
\right) \geq \left( k+1\right) \cdot d\cdot \left( 1-\gamma \right) \cdot
\log \left( n\right) 
\end{equation*}%
holds asymptotically. We obtain that, for all $k>\frac{1}{d}$, the language $L_{k}$ does not belong to $\mathcal{%
NREG}_{\left\lfloor \frac{k}{d}\right\rfloor }$. This implies that $\left\{
L_{k}\right\} _{k\geq 0}$ is high in the nondeterministic pebble hierarchy.
\end{proof}

\section{Context-free languages that are hard-to-cross\label{Section7}}

Corollary \ref{corolarioInductivo} reduces the problem of proving $\mathcal{%
NL}\neq \log \mathcal{CFL}$ to that of constructing a hard-to-cross
context-free language.

Let $L\subseteq \Sigma ^{\ast }$ be a context-free language. We study the%
\textit{\ crossing complexity} of separating $L_{1}\cap \mathcal{H}%
_{1}^{\Sigma }$ from $co$-$L_{1}\cap \mathcal{H}_{1}^{\Sigma }$.

\subsection{Entropy and state complexity of context-free languages}

Two-way nondeterministic finite state automata (2NFAs) are the same as
nondeterministic pebble automata provided with zero pebbles (i.e., 2NFAs are
nondeterministic zero-pebble automata).

\begin{definition}
Let $L\subseteq \Sigma ^{\ast }$ be a context-free language and $\#\notin
\Sigma $.

\begin{enumerate}
\item Let $n\geq 2$, we use $POS_{n}\left( L\right) $ to denote the finite
set:%
\begin{equation*}
\left\{
\begin{array}{c}
\boldsymbol{\varepsilon}\# r_1 \# \cdots \#
r_{\lfloor \log(n) \rfloor} : \\[4pt]

\boldsymbol{\varepsilon} \in
\{0,1\}^{\lfloor \log(n) \rfloor},
\quad
r_1,\ldots,r_{\lfloor \log(n) \rfloor}
\in \Sigma^{\leq n}; \\[4pt]

\text{and }
\bigodot_{i:\boldsymbol{\varepsilon}[i]=1} r_i \in L
\end{array}
\right\}.
\end{equation*}

\item We use $NEG_{n}\left( L\right) $ to denote the finite set:%
\begin{equation*}
\left\{
\begin{array}{c}
\boldsymbol{\varepsilon}\# r_1 \# \cdots \#
r_{\lfloor \log(n) \rfloor} : \\[4pt]

\boldsymbol{\varepsilon} \in
\{0,1\}^{\lfloor \log(n) \rfloor},
\quad
r_1,\ldots,r_{\lfloor \log(n) \rfloor}
\in \Sigma^{\leq n}; \\[4pt]

\text{and }
\bigodot_{i:\boldsymbol{\varepsilon}[i]=1} r_i \notin L
\end{array}
\right\}.
\end{equation*}
\end{enumerate}
\end{definition}

\begin{definition}
Let $\mathcal{A}$ be a 2NFA that separates $POS_{n}\left( L\right) $ from $%
NEG_{n}\left( L\right) $. Let $q$ be an internal state of $\mathcal{A}$. We
say that $q$ is a \textit{crossing state} if it is visited while $\mathcal{A}
$ processes an input from $POS_{n}\left( L\right) \cup NEG_{n}\left(
L\right) $ and the automaton passes through the cell containing the leftmost
occurrence of $\#$. The \textit{crossing size} of $\mathcal{A}$ is the
number of its crossing states. We use $\left\Vert \mathcal{A}\right\Vert _{c}
$ to denote the crossing size of $\mathcal{A}$.
\end{definition}

\begin{remark}
The crossing states are the internal states that are used to carry
information from one side (either $\mathbf{\varepsilon }$ or $r_{1}\#\cdots
\#r_{\left\lfloor \log \left( n\right) \right\rfloor }$) to the other.
\end{remark}

\begin{definition}
Let $L\subseteq \Sigma ^{\ast }$. The nondeterministic crossing complexity
of $L$ is defined by:%
\begin{equation*}
2ncc\left( L,n\right) =\min \left\{ \left\Vert \mathcal{A}\right\Vert _{c}%
\text{: }\mathcal{A}\text{ is a 2NFA that separates }POS_{n}\left( L\right) 
\text{ from }NEG_{n}\left( L\right) \right\} .
\end{equation*}
\end{definition}

We have:

\begin{theorem}
Let $L\subseteq \Sigma ^{\ast }$. Let $\mathcal{A}_{1}$ be a
nondeterministic pebble automaton separating $L_{1}\cap \mathcal{H}%
_{1}^{\Sigma }$ from $co$-$L_{1}\cap \mathcal{H}_{1}^{\Sigma }$. Suppose the
inequality 
\begin{equation*}
H\left( Y_{\mathcal{A}_{1}}\left( \mathcal{H}_{1}^{\Sigma },n\right) \right)
\leq f\left( n\right)
\end{equation*}%
holds for infinitely many values of $n.$ Then, there exists $C$ such that
the following inequality holds for infinitely many values of $n$:%
\begin{equation*}
2ncc\left( L,n\right) \leq 2^{f\left( n\right) }\cdot C\cdot \log \left(
n\right) .
\end{equation*}
\end{theorem}

\begin{proof}
Let $\mathcal{A}_{1}$ be a nondeterministic pebble automaton separating $%
L_{1}\cap \mathcal{H}_{1}^{\Sigma }$ from $co$-$L_{1}\cap \mathcal{H}%
_{1}^{\Sigma }$ and such that the inequality%
\begin{equation*}
H\left( Y_{\mathcal{A}_{1}}\left( \mathcal{H}_{1}^{\Sigma },n\right) \right)
\leq f\left( n\right)
\end{equation*}%
holds for infinitely many values of $n$. Let us fix such a $n$. Let $%
\mathcal{C}_{\mathcal{A}_{1}}\left( \mathcal{H}_{1}^{\Sigma },n\right) $ be
the set of coding configurations that are visited by $\mathcal{A}_{1}$ when
this automaton is processing an input from $\mathcal{H}_{1}^{\Sigma }\left(
3\cdot \left\lfloor \log \left( n\right) \right\rfloor \cdot \left(
n+1\right) \right) $. And let $\mathcal{CC}_{\mathcal{A}_{1}}\left( \mathcal{%
H}_{1}^{\Sigma },n\right) $ be the subset constituted by all the
configurations visited by the automaton when it is asked to compute a 
\textit{local bit} $w_{i}\left[ 1\right] $. We construct a 2NFA $\mathcal{M}%
_{n}$ such that:

\begin{enumerate}
\item $\mathcal{M}_{n}$ separates $POS_{n}\left( L\right) $ from $%
NEG_{n}\left( L\right) $.

\item The set of internal states of $\mathcal{M}_{n}$ is equal to: 
\begin{equation*}
\mathcal{C}_{\mathcal{A}_{1}}\left( \mathcal{H}_{1}^{\Sigma },n\right)
\times \left\{ 0,\ldots ,3\cdot \left\lfloor \log \left( n\right)
\right\rfloor \cdot \left( n+1\right) \right\} ^{2}\times \left\{ 1,...,\log
\left( n\right) \right\} ^{2}.
\end{equation*}

\item The set of crossing states of $\mathcal{M}_{n}$ is included in 
\begin{equation*}
\mathcal{CC}_{\mathcal{A}_{1}}\left( \mathcal{H}_{1}^{\Sigma },n\right)
\times \left\{ \left( \left\lfloor \log \left( n\right) \right\rfloor
,0\right) \right\} \times \left\{ 0,...,\log \left( n\right) \right\} ^{2}.
\end{equation*}
\end{enumerate}

Suppose $\mathcal{A}_{1}$ is provided with $m$-pebbles. Recall that an
instantaneous configuration of $\mathcal{A}_{1}$ is a tuple $\left( h,\left(
c_{1},...,c_{m},q\right) \right) $ such that $h,c_{i}\in \left\{
0,...,3\cdot \left\lfloor \log \left( n\right) \right\rfloor \cdot \left(
n+1\right) \right\} $ and $q$ is an internal state of $\mathcal{A}_{1}$. 

Suppose:

\begin{itemize}
\item $\mathcal{A}_{1}$ is in the instantaneous configuration $\left(
h,\left( c_{1},...,c_{m},q\right) \right) $ scanning the local pair $%
w_{i}\#_{1}r_{i}$. 
\end{itemize}

Then:

\begin{itemize}
\item $\mathcal{M}_{n}$ is in the internal state $\left(
c_{1},...,c_{m},q,\left( h,0\right) ,\left( i,0\right) \right) $. 
\end{itemize}

Moreover, we have:

\begin{itemize}
\item Let $\left( 0,\left( 0,\ldots ,0,q_{0}\right) \right) $ be the initial
instantaneous configuration of $\mathcal{A}_{1}$. The tuple 
\begin{equation*}
\left( 0,\ldots ,0,q_{0},\left( 0,0\right) ,\left( 0,0\right) \right) 
\end{equation*}%
is the initial state of $\mathcal{M}_{n}$.

\item The set of accepting states of $\mathcal{M}_{n}$ is the set%
\begin{equation*}
\left\{ \left( \overrightarrow{c},q,\left( l,t\right) \left( i,j\right)
\right) :\left( l,\left( \overrightarrow{c},q\right) \right) \text{ is an
accepting instantaneous configuration of }\mathcal{A}_{1}\right\} .
\end{equation*}
\end{itemize}

It only remains to define the transition function of $\mathcal{M}_{n}$. Let 
\begin{equation*}
s\left( w,r\right) =w_{1}\#_{1}r_{1}\#_{1}\cdots \#_{1}w_{\left\lfloor \log
\left( n\right) \right\rfloor }\#_{1}r_{\left\lfloor \log \left( n\right)
\right\rfloor }\#_{1}^{s}
\end{equation*}%
be an input of $\mathcal{A}_{1}$, and let 
\begin{equation*}
S\left( w,r\right) =w_{1}\left[ 1\right] \cdots w_{\left\lfloor \log \left(
n\right) \right\rfloor }\left[ 1\right] \#r_{1}\#\cdots \#r_{\left\lfloor
\log \left( n\right) \right\rfloor }
\end{equation*}%
be the corresponding input of $\mathcal{M}_{n}$. The latter automaton
simulates, on $S\left( w,r\right) ,$ the computation of $\mathcal{A}_{1}$ on 
$s\left( w,r\right) $. Let $\left( \overrightarrow{c},q,\left( l,t\right)
\left( i,j\right) \right) $ be an internal state of $\mathcal{M}_{n}$. This
automaton uses the coordinate $l$ to track the location of its input head.
The coordinate $t$ is used as an auxiliary counter. On the other hand, $%
\mathcal{M}_{n}$ uses the last two components to track the factor being
processed. Thus, suppose the input head of $\mathcal{A}_{1}$ is scanning the 
$j$-th letter of $\#_{1}r_{i}$, and suppose that the current coding
configuration of this pebble automaton is equal to\ $\left( \overrightarrow{c%
},q\right) $. We assume that $\mathcal{M}_{n}$ is scanning the $j$-th letter
of $\#r_{i}$ and that the internal state is equal to 
\begin{equation*}
\left( \overrightarrow{c},q,\left( \left\vert r_{1}\#\cdots
\#r_{i-1}\#\right\vert +j,0\right) ,\left( i,0\right) \right) .
\end{equation*}

Notice that $\mathcal{M}_{n}$ can use the auxiliary counter represented by $%
t $ to:

\begin{itemize}
\item check if the input head position matches with the position of a
specific \textit{virtual} pebble (i.e., matches $\overrightarrow{c}\left[ i%
\right] $ for some $i\leq k$), and

\item reset the contents of the first counter (i.e., the value of $l$) after
this check.
\end{itemize}

Suppose $\mathcal{A}_{1}$ chooses to continue with the processing of $r_{i}$%
. The automaton $\mathcal{M}_{n}$ simulates this phase of the computation
using the internal states in the set 
\begin{equation*}
\mathcal{C}_{\mathcal{A}_{1}}\left( \mathcal{H}_{1}^{\Sigma },n\right)
\times \left\{ 1,\ldots ,3\cdot \left\lfloor \log \left( n\right)
\right\rfloor \cdot \left( n+1\right) \right\} ^{2}\times \left\{ \left(
i,0\right) \right\} .
\end{equation*}%
Notice that all those internal states of $\mathcal{M}_{n}$ are non-crossing
states.

Now suppose that $\left( \overrightarrow{c},q\right) \in \mathcal{CC}_{%
\mathcal{A}_{1}}\left( \mathcal{H}_{1}^{\Sigma },n\right) $ and that $%
\mathcal{A}_{1}$ chooses to scan the bit $w_{i}\left[ 1\right] $. Then, the
automaton $\mathcal{M}_{n}$ has to move to the left in order to scan the $i$%
-th cell. Notice that, during this phase of the computation, $\mathcal{M}_{n}
$ can set $t,j=0.$ Then, $\mathcal{M}_{n}$ can simulate this phase of the
computation using the set 
\begin{equation*}
\left\{ 0,\ldots ,3\cdot \left\lfloor \log \left( n\right) \right\rfloor
\cdot \left( n+1\right) \right\} \times \left\{ 0\right\} \times \left\{
\left( \overrightarrow{c},q\right) \right\} \times \left\{ \left( i,0\right)
\right\} .
\end{equation*}%
It can be verified that:

\begin{enumerate}
\item $\mathcal{M}_{n}$ separates $POS_{n}\left( L\right) $ from $%
NEG_{n}\left( L\right) $.

\item The set of crossing states of $\mathcal{M}_{n}$ is included in%
\begin{equation*}
\left\{ \left( \left\lfloor \log \left( n\right) \right\rfloor +1,0\right)
\right\} \times \mathcal{CC}_{\mathcal{A}_{1}}\left( \mathcal{H}_{1}^{\Sigma
},n\right) \times \left\{ 1,...,\left\lfloor \log \left( n\right)
\right\rfloor \right\} \times \left\{ 0\right\} .
\end{equation*}
\end{enumerate}

Let $C=\left\vert Q\right\vert $. The following inequalities hold for
infinitely many values of $n$:%
\begin{equation*}
2ncc\left( L,n\right) \leq \left\vert \mathcal{CC}_{\mathcal{A}_{1}}\left( 
\mathcal{H}_{1}^{\Sigma },n\right) \right\vert \cdot C\cdot \log ^{2}\left(
n\right) \leq 2^{f\left( n\right) }\cdot C\cdot \log \left( n\right) .
\end{equation*}
\end{proof}

We obtain the following corollary:

\begin{corollary}
Let $e,d>0$, and suppose that the inequality $2ncc\left( L,n\right) \geq
e\cdot n^{d}$ holds asymptotically. Then, for all $c<d$, all
nondeterministic pebble automata $\mathcal{A}_{1}$ separating $L_{1}\cap 
\mathcal{H}_{1}^{\Sigma }$ from $co$-$L_{1}\cap \mathcal{H}_{1}^{\Sigma }$,
and all $\gamma >0$, the following inequality holds asymptotically:%
\begin{equation*}
H\left( Y_{\mathcal{A}_{1}}\left( \mathcal{H}_{1}^{\Sigma },n\right) \right)
\geq c\cdot \left( 1-\gamma \right) \cdot \log \left( n\right) .
\end{equation*}
\end{corollary}

\begin{proof}
Let $c<d$. Suppose that the inequality%
\begin{equation*}
H\left( Y_{\mathcal{A}_{1}}\left( \mathcal{H}_{1}^{\Sigma },n\right) \right)
\geq c\cdot \left( 1-\gamma \right) \cdot \log \left( n\right)
\end{equation*}%
does not hold asymptotically for all $\gamma >0$. Then, there exists $\gamma
>0$ such that the inequality%
\begin{equation*}
H\left( Y_{\mathcal{A}_{1}}\left( \mathcal{H}_{1}^{\Sigma },n\right) \right)
<d\cdot \left( 1-\gamma \right) \cdot \log \left( n\right)
\end{equation*}%
holds for infinitely many values of $n.$ Let $n$ be one such value. The
random variable $Y_{\mathcal{A}_{1}}\left( \mathcal{H}_{1}^{\Sigma
},n\right) $ is uniformly distributed over $\mathcal{CC}_{\mathcal{A}%
_{1}}\left( \mathcal{H}_{1}^{\Sigma },n\right) $. Then:%
\begin{equation*}
\left\vert \mathcal{CC}_{\mathcal{A}_{1}}\left( \mathcal{H}_{1}^{\Sigma
},n\right) \right\vert <2^{\log \left( n^{c}\right) }.
\end{equation*}

We get that:

\begin{enumerate}
\item There exists $C$ such that $2ncc\left( L,n\right) \leq n^{c}\cdot
C\cdot \log \left( n\right) $ holds for infinitely many values of $n$.

\item The inequality $2ncc\left( L,n\right) \geq e\cdot n^{d}$ holds
asymptotically.
\end{enumerate}

We obtain a contradiction, the corollary follows.
\end{proof}

\begin{definition}
Let $L\in \mathcal{CFL}$. We say that $L$ is crossing-hard if there exists $%
c,d>0$ such that, for all $n\geq 1$, the following inequality holds: 
\begin{equation*}
2ncc\left( L,n\right) >c\cdot n^{d}.
\end{equation*}
\end{definition}

We have:

\begin{corollary}
\label{CorollaryCorssing Hard}The following assertions hold:

\begin{enumerate}
\item Let $L\in \mathcal{CFL}$ be crossing-hard. Then, the language $L$ is
hard-to-cross.

\item Let $L\in \mathcal{CFL}$. If $L$ is crossing-hard the sequence $%
\left\{ L_{k}\right\} _{k\geq 0}$ is high in the nondeterministic pebble
hierarchy.

\item If there exists a context-free language that is crossing-hard the
separation $\mathcal{NL}\neq \log \mathcal{CFL}$ holds.
\end{enumerate}
\end{corollary}

\section{Rectangles, digraph accessibility, and the language RA\label%
{Section8}}

Corollary \ref{CorollaryCorssing Hard} reduces the problem of proving the
separation $\mathcal{NL}\neq \log \mathcal{CFL}$ to that of constructing a
context-free language that is crossing-hard. We construct a context-free
language, denoted RA, and show that it is crossing-hard. This language is closely
related to the digraph accessibility problem.

\begin{definition}
Let $G$ be a digraph. We say that $G$ is a rectangle of depth $k$ and width $%
n$ if and only if:

\begin{enumerate}
\item $V\left( G\right) =\left\{ 1,\ldots ,n\right\} \times \left\{ 1,\ldots
,k+1\right\} .$

\item If $\left( \left( i,j\right) ,\left( r,t\right) \right) \in E\left(
G\right) $ the equality $t=j+1$ holds.
\end{enumerate}
\end{definition}

Rectangles are \textit{layered digraphs} with edges directed strictly
between successive layers (levels). It is worth recalling that the
restriction of digraph accessibility to the class of layered digraphs is $%
\mathcal{NL}$-complete; see \cite{Jones}. One can verify that the
restriction of this latter problem to the class of rectangles is also $%
\mathcal{NL}$-complete.

\begin{definition}
Let $G$ be a rectangle of depth $k$. We say that $G$ is crossable if there
exists a path from the first level of $G$ (i.e., the set $\left\{
1,...,n\right\} \times \left\{ 1\right\} $) to the level $k+1$.
\end{definition}

\begin{definition}
Let $G$ be a rectangle of depth $k$. We encode $G$ as follows:

\begin{itemize}
\item Let $s<k+1$ and let%
\begin{equation*}
l\left( s\right) =\left\{ \left( \left( i_{1},s\right) ,\left(
j_{1},s+1\right) \right) ,...,\left( \left( i_{m_{s}},s\right) ,\left(
j_{m_{s}},s+1\right) \right) \right\} 
\end{equation*}%
be a list (which may include repetitions) of the edges that go from level $s$
to level $s+1.$ We define:%
\begin{equation*}
r_{s,l\left( s\right) }=1^{i_{1}}01^{j_{1}}0\cdots
01^{i_{m_{s}}}01^{j_{m_{s}}}.
\end{equation*}

\item We set $w_{G,l}=r_{1,l\left( 1\right) }\$\cdots \$r_{k,l\left(
k\right) }\$$.
\end{itemize}
\end{definition}

Let $G$ be rectangle of depth $k.$ Suppose $w_{G,l}=r_{1,l\left( 1\right)
}\$\cdots \$r_{k,l\left( k\right) }\$$. Notice that $R$ is the concatenation
of $k$ rectangles of depth $1$, the rectangles encoded by the strings $%
r_{1,l\left( 1\right) }\$,\ldots ,r_{k-1,l\left( k-1\right) }\$$, and $%
r_{k,l\left( k\right) }\$$.

\begin{remark}
Let $n$ be a large integer. Let $G$ be a rectangle of depth $1$ and width $%
\left\lfloor n^{\frac{1}{4}}\right\rfloor $. This rectangle contains no more
than $n^{\frac{1}{2}}$ edges. Furthermore, the code for each of these edges
is no longer than $2n^{\frac{1}{4}}+2.$ Then there exists a code-word $%
w_{G,l}$ such that $\left\vert w_{G,l}\right\vert \leq n$, that is: we can
think of $w\in \left\{ 0,1\right\} ^{n}$ as the code of a rectangle of depth 
$1$ and width $\left\lfloor n^{\frac{1}{4}}\right\rfloor $.
\end{remark}

\begin{definition}
We use RA to denote the following context-free language%
\begin{equation*}
\left\{ w_{G,l}\in \left\{ 0,1,\$\right\} ^{\ast }:G\text{ is crossable}%
\right\} .
\end{equation*}
\end{definition}

We have:

\begin{theorem}
RA is a context-free language.
\end{theorem}

\begin{proof}
Let $w_{1}\$\cdot \cdots \$w_{k}\$$ be an instance of RA. This instance
belongs to RA if and only if there exists factors $1^{i_{1}}01^{j_{1}},%
\ldots ,1^{i_{k}}01^{j_{k}}$ such that:

\begin{enumerate}
\item $1^{i_{1}}01^{j_{1}}$ is a factor of $w_{1},$\ldots $,$ $%
1^{i_{k-1}}01^{j_{k-1}}$ is a factor of $w_{k-1}$, and $1^{i_{k}}01^{j_{k}}$
is a factor of $w_{k}$.

\item $j_{1}=i_{2},\ldots ,j_{k-2}=i_{k-1}$, and $j_{k-1}=i_{k}$.
\end{enumerate}

It is easy to build a nondeterministic pushdown automaton that guesses those
factors and checks all those equalities.
\end{proof}

We set RA$_{\infty }=L_{G_{0}}$, we have:

\begin{enumerate}
\item RA$_{\infty }$ is an upper bound for $\left\{ \text{RA}_{k}\right\}
_{k\geq 0}$.

\item If RA is crossing-hard the sequence $\left\{ \text{RA}_{k}\right\}
_{k\geq 0}$ is high in the nondeterministic pebble hierarchy and the
language RA$_{\infty }$ belongs to $\mathcal{NL}-\log \mathcal{CFL}$.
\end{enumerate}

\subsection{Decorated versions of RA}

From now, we study the crossing complexity of RA. We aim to show that RA is
crossing-hard. Let us observe that this problem is a question about the 
\textit{communication complexity} of 2NFAs.  

\begin{definition}
Let $m\geq 1$. We define and injective function from $\left\{ 1,\ldots ,m\right\} $
to $\left\{ 0,1\right\} ^{m}$. We define this injection as follows:

\begin{enumerate}
\item If $m$ is even, we set%
\begin{equation*}
\left[ i\right] _{m}=\left\{ 
\begin{array}{c}
\left( 10\right) ^{i}0^{m-2i},\text{ if }i\leq \frac{m}{2} \\ 
\left( 11\right) ^{i-\frac{m}{2}}\left( 10\right) ^{\ast },\text{ otherwise}%
\end{array}%
\right. .
\end{equation*}

\item If $m$ is odd, we set%
\begin{equation*}
\left[ i\right] _{m}=\left\{ 
\begin{array}{c}
\left( 10\right) ^{i}0^{m-2i},\text{ if }i<\frac{m}{2} \\ 
\left( 11\right) ^{i-\left\lfloor \frac{m}{2}\right\rfloor }\left( 10\right)
^{\ast }0,\text{ if }\frac{m}{2}<i<m \\ 
1^{m}\text{, if }i=m%
\end{array}%
\right. .
\end{equation*}
\end{enumerate}
\end{definition}

The next proposition is straightforward.

\begin{proposition}
\label{proposicion tecnica boba}Suppose $m$ is even. Then:

\begin{enumerate}
\item $\left[ \frac{m}{2}+1\right] _{m}=11\left( 10\right) ^{\frac{m}{2}-1}$.

\item $\left[ \frac{m}{2}+2\right] _{m}=1111\left( 10\right) ^{\frac{m}{2}%
-2}.$

\item For all $i\neq \frac{m}{2}+1,\left[ \frac{m}{2}+2\right] _{m}$, the
string $\left[ m\right] _{i}$ does not belong to $11\left( 10\right) ^{\ast
}\cup 1111\left( 10\right) ^{\ast }$.
\end{enumerate}
\end{proposition}

\begin{definition}
We define:

\begin{itemize}
\item The set $POS_{n}^{\ast }\left( \text{RA}\right) $ is equal to:%
\begin{equation*}
\left\{
\begin{array}{c}
\boldsymbol{\varepsilon}
\# [1]_{\lfloor \log(n) \rfloor}\$ r_1
\# \cdots \#
[\log(n)]_{\lfloor \log(n) \rfloor}\$
r_{\lfloor \log(n) \rfloor}
: \\[4pt]

\boldsymbol{\varepsilon}
\in \{0,1\}^{\lfloor \log(n) \rfloor},
\quad
r_1,\ldots,r_{\lfloor \log(n) \rfloor}
\in \{0,1\}^{\leq n}; \\[4pt]

\text{and }
\bigodot_{\boldsymbol{\varepsilon}[i]=1} r_i\$
\in \mathrm{RA}
\end{array}
\right\}.
\end{equation*}

\item The set $NEG_{n}^{\ast }\left( \text{RA}\right) $ is equal to: 
\begin{equation*}
\left\{
\begin{array}{c}
\boldsymbol{\varepsilon}
\# [1]_{\lfloor \log(n) \rfloor}\$ r_1
\# \cdots \#
[\lfloor \log(n) \rfloor]_{\lfloor \log(n) \rfloor}\$
r_{\lfloor \log(n) \rfloor}
: \\[4pt]

\boldsymbol{\varepsilon}
\in \{0,1\}^{\lfloor \log(n) \rfloor},
\quad
r_1,\ldots,r_{\lfloor \log(n) \rfloor}
\in \{0,1\}^{\leq n}; \\[4pt]

\text{and }
\bigodot_{\boldsymbol{\varepsilon}[i]=1} r_i\$
\notin \mathrm{RA}
\end{array}
\right\}.
\end{equation*}
\end{itemize}
\end{definition}

\begin{remark}
It is important to observe that the elements of $POS_{n}^{\ast }\left( \text{%
RA}\right) \cup NEG_{n}^{\ast }\left( \text{RA}\right) $ are decorated
versions of the elements of $POS_{n}\left( \text{RA}\right) \cup
NEG_{n}\left( \text{RA}\right) $. We will show that those
decorations do not increase the crossing complexity.
\end{remark}

Let $\mathcal{M}$ be a 2NFA that separates $POS_{n}^{\ast }\left( \text{RA}%
\right) $ from $NEG_{n}^{\ast }\left( \text{RA}\right) $. The crossing
states of $\mathcal{M}$ are the internal states visited when the automaton
passes through the cell containing the leftmost occurrence of $\#.$ We
define $\left\Vert \mathcal{M}\right\Vert _{c}$ accordingly and we define $%
ncc\left( POS_{n}^{\ast }\left( \text{RA}\right) ,NEG_{n}^{\ast }\left( 
\text{RA}\right) \right) $ as%
\begin{equation*}
\min \left\{ \left\Vert \mathcal{M}\right\Vert _{c}:\mathcal{M}\text{
separates }POS_{n}^{\ast }\left( \text{RA}\right) \text{ from }NEG_{n}^{\ast
}\left( \text{RA}\right) \right\} .
\end{equation*}

We have:

\begin{lemma}
For all $n\geq 1$, the following inequality holds:%
\begin{equation*}
ncc\left( POS_{n}\left( \text{RA}\right) ,NEG_{n}\left( \text{RA}\right)
\right) \geq ncc\left( POS_{n}^{\ast }\left( \text{RA}\right) ,NEG_{n}^{\ast
}\left( \text{RA}\right) \right) .
\end{equation*}
\end{lemma}

\begin{proof}
Let $\mathcal{M}$ be a 2NFA that separates $POS_{n}\left( \text{RA}\right) $
from $NEG_{n}\left( \text{RA}\right) $. This automaton can be used to
separate $POS_{n}^{\ast }\left( \text{RA}\right) $ from $NEG_{n}^{\ast
}\left( \text{RA}\right) $. It suffices if, on the input $\mathbf{%
\varepsilon }\#\left[ 1\right] _{\left\lfloor \log \left( n\right)
\right\rfloor }\$r_{1}\#\cdots \#\left[ \log \left( n\right) \right]
_{\left\lfloor \log \left( n\right) \right\rfloor }\$r_{\left\lfloor \log
\left( n\right) \right\rfloor }$, the automaton ignores the infixes 
\begin{equation*}
\#\left[ 1\right] _{\left\lfloor \log \left( n\right) \right\rfloor
}\$,\ldots ,\#\left[ \left\lfloor \log \left( n\right) -1\right\rfloor %
\right] _{\left\lfloor \log \left( n\right) \right\rfloor }\$,\text{ and }%
\left[ \left\lfloor \log \left( n\right) \right\rfloor \right]
_{\left\lfloor \log \left( n\right) \right\rfloor }.
\end{equation*}%
The lemma follows from this.
\end{proof}

We show, in Section \ref{Section9}, that there exist $c,d>0$ such that, for
all $n$ large enough, the following inequality holds:

\begin{equation*}
2ncc\left( POS_{n}^{\ast }\left( \text{RA}\right) ,NEG_{n}^{\ast }\left( 
\text{RA}\right) \right) \geq c\cdot n^{d}.
\end{equation*}%
This implies that RA is crossing-hard and this also implies the separation $%
\mathcal{NL}\neq \log \mathcal{CFL}$.

\section{On the crossing complexity of pattern matching}

We are interested in the crossing complexity of digraph accessibility (the crossing complexity of the language RA). This leads us to study the crossing complexity of pattern
matching:

\begin{enumerate}
\item We show that a promise problem related to pattern matching is crossing
hard.

\item We show that the crossing complexity of this pattern matching problem
is at most as high as the crossing complexity of RA.
\end{enumerate}

Let $i\geq 1$. Recall that $\left\langle i\right\rangle $ denotes the binary
code of $i-1$.

\begin{definition}
Suppose $i\leq 2^{m+1}.$ We define:%
\begin{equation*}
\left\langle i\right\rangle _{m}=0^{m-\left\vert \left\langle i\right\rangle
\right\vert }\left\langle i\right\rangle .
\end{equation*}
\end{definition}

\begin{definition}
Let $n\geq 2$, we use $PM_{n}$ to denote the following promise problem:

\begin{itemize}
\item Input:$\ \mathbf{\varepsilon }\#1^{n_{1}}\#\cdots \#1^{n_{m}}$, where $%
\mathbf{\varepsilon }\in \left\{ 0,1\right\} ^{\left\lfloor \log \left(
n\right) \right\rfloor }$; $m\leq 2^{4\cdot \left\lfloor \log \left(
n\right) \right\rfloor }$; and $n_{1},\ldots ,n_{m}\leq 2^{\left\lfloor \log
\left( n\right) \right\rfloor }$.

\item Problem: decide whether there exists $i\leq m$ such that $\mathbf{%
\varepsilon }=\left\langle n_{i}\right\rangle _{\left\lfloor \log \left(
n\right) \right\rfloor }.$
\end{itemize}
\end{definition}

We use $PM_{n}^{+}$ to denote the set of positive instances of $PM_{n}$. We
use $PM_{n}^{-}$ to denote the set of negative instances.

\begin{remark}
From now on, we use this notation with any promise problem under
consideration, (i.e., \thinspace $L^{+}$ and $L^{-}$ denote the sets of
positive and negative instances of the promise problem $L$).
\end{remark}

We define $2ncc\left( PM_{n}^{+},PM_{n}^{-}\right) $ in the usual way. We
have:

\begin{theorem}
For all $n\geq 2,$ the following inequality holds:%
\begin{equation*}
2ncc\left( PM_{n}^{+},PM_{n}^{-}\right) \geq \left( \frac{n}{2}\right) ^{%
\frac{1}{2}}.
\end{equation*}
\end{theorem}

\begin{proof}
Let $n\geq 2$ and let $w=\mathbf{\varepsilon }\#1^{n_{1}}\#\cdots
\#1^{n_{k}} $ be an instance of $PM_{n}.$ We use $S\left( w\right) $ to
denote the set 
\begin{equation*}
\bigcup_{i \leq k} \{n_k\}.
\end{equation*}
which we call the \textit{support} of $w.$

Let $\mathcal{M}_{n}$ be a 2NFA that solves the problem $PM_{n}$ and let $%
Q_{c}$ be its set of crossing states. Let $f_{w,\mathcal{M}%
_{n}}:Q_{c}\rightarrow 2^{Q_{c}}$ be the function defined by:

Let $p\in Q_{c}$, suppose the input head of $\mathcal{M}_{n}$ passes from
left ($\mathbf{\varepsilon }$) to right ($1^{n_{1}}\#\cdots \#1^{n_{k}}$) in
the crossing state $p,$ the set of crossing states that can be visited in
the next pass (from right to left) is equal to $f_{w,\mathcal{M}_{n}}\left(
p\right) $.

Suppose there exist two inputs $w=\mathbf{\varepsilon }\#1^{n_{1}}\#\cdots
\#1^{n_{k}}$ and $u=\mathbf{\delta }\#1^{s_{1}}\#\cdots \#1^{s_{t}}$ such
that:

\begin{itemize}
\item $S\left( w\right) \neq S\left( u\right) ,$ and

\item $f_{w,\mathcal{M}_{n}}=$ $f_{u,\mathcal{M}_{n}}$.
\end{itemize}

Let $j\notin S\left( w\right) \cap S\left( u\right) $ and let $\mathbf{%
\gamma }=\left\langle j\right\rangle _{\left\lfloor \log \left( n\right)
\right\rfloor }$. Set $w^{\ast }=\mathbf{\gamma }\#1^{n_{1}}\#\cdots
\#1^{n_{k}}$ and $u^{\ast }=\mathbf{\gamma }\#1^{s_{1}}\#\cdots \#1^{s_{t}}.$
Let us consider these two inputs. We have:

\begin{enumerate}
\item $w^{\ast }\in PM_{n}^{+}\Leftrightarrow u^{\ast }\notin PM_{n}^{+}$

\item The automaton accepts both inputs or rejects both.
\end{enumerate}

This is a contradiction. We conclude that there exists an injective function
from $\mathcal{P}\left( \left\{ 1,...,2^{\left\lfloor \log \left( n\right)
\right\rfloor }\right\} \right) $ into $\left( 2^{Q_{c}}\right) ^{Q_{c}}$.
This entails the inequality 
\begin{equation*}
\left\Vert \mathcal{M}_{n}\right\Vert _{c}\geq \left( \frac{n}{2}\right) ^{%
\frac{1}{2}}.
\end{equation*}%
The theorem follows.
\end{proof}

In the next section, we show that a 2NFA separating $POS_{n}^{\ast }\left( 
\text{RA}\right) $ from $NEG_{n}^{\ast }\left( \text{RA}\right) $ requires
as many states as a 2NFA solving the problem $PM_{n^{\frac{1}{4}}}$. This
implies that RA is crossing-hard.

\section{The language RA is crossing-hard\label{Section9}}

In this section, the last section of this work, we prove that RA is
crossing-hard. We get as a corollary the separation $\mathcal{NL}\neq \log 
\mathcal{CFL}$.

\subsection{Decorated versions of pattern matching}

\begin{definition}
We use the following notation:

\begin{enumerate}
\item Let $x,y\in \left\{ 0,1\right\} ,$ we use $x\oplus y$ to denote the
XOR of $x$ and $y$.

\item Let $\mathbf{\varepsilon }=\varepsilon _{1}\cdots \varepsilon _{m}\in
\left\{ 0,1\right\} ^{m}$. We define: 
\begin{equation*}
\mathbf{\varepsilon }^{\oplus }=\varepsilon _{1}\left( 1\oplus \varepsilon
_{1}\right) \cdots \varepsilon _{m}\left( 1\oplus \varepsilon _{m}\right) .
\end{equation*}
\end{enumerate}
\end{definition}

Next proposition is straightforward.

\begin{proposition}
Let $\mathbf{\varepsilon }$ and $\mathbf{\varepsilon }^{\oplus }$ be as
above. The pattern $111$ does not occur in $\mathbf{\varepsilon }^{\oplus }$.
\end{proposition}

Consider the following \textit{decorated version }of the problem $PM_{n^{%
\frac{1}{4}}}$.

\begin{definition}
Let $n\geq 1$. Let $X,Y\in \left\{ 0,1,\#,\$\right\} ^{\ast }$ be two
strings such that $X$ avoids the pattern $\#11\left( 10\right) ^{\ast }\$$
and $Y$ avoids the pattern $\#1111\left( 10\right) ^{\ast }\$$. We use $%
LOC_{n}\left( X,Y\right) $ to denote the following promise problem.

\begin{itemize}
\item Input: $\mathbf{\varepsilon }^{\oplus }111X\#11\left( 10\right) ^{\ast
}\$\mathbf{1}^{n_{1}}\mathbf{0\cdots 01}^{n_{m}}\#1111\left( 10\right)
^{\ast }\$Y$,

where $\mathbf{\varepsilon }\in \left\{ 0,1\right\} ^{\left\lfloor \frac{%
\log \left( n\right) }{4}\right\rfloor }$; $m\leq 2^{\left\lfloor \log
\left( n\right) \right\rfloor }$; and $n_{1},...,n_{m}\leq 2^{\left\lfloor 
\frac{\log \left( n\right) }{4}\right\rfloor }$.

\item Problem: decide whether there exists an odd $i\leq m$ such that $%
\left\langle n_{i}\right\rangle _{\left\lfloor \frac{\log \left( n\right) }{2%
}\right\rfloor }=\mathbf{\varepsilon }$.
\end{itemize}
\end{definition}

\begin{definition}
Let $\mathcal{M}_{n}$ be a two-way automaton solving the problem $%
LOC_{n}\left( X,Y\right) $. Let $q$ be an internal state of $\mathcal{M}_{n}$%
. We say that $q$ is a crossing state if $q$ is visited while the input head
is scanning the leftmost occurrence of the pattern $111$.
\end{definition}

It is worth noting that instances of $LOC_{n}\left( X,Y\right) $ are nothing
more than instances of $PM_{n^{\frac{1}{4}}}$ but decorated with an infix
and a suffix that can be easily recognized. We have:

\begin{lemma}
Let $X$ and $Y$ be as above. For all $n\geq 1,$ the following inequality
holds:%
\begin{equation*}
2ncc\left( LOC_{n}^{+}\left( X,Y\right) ,LOC_{n}^{-}\left( X,Y\right)
\right) \geq 2ncc\left( PM_{n^{\frac{1}{4}}}^{+},PM_{n^{\frac{1}{4}%
}}^{-}\right) .
\end{equation*}
\end{lemma}

\begin{proof}
Let $n\geq 1$ and let $\mathcal{M}_{n}$ be a nondeterministic pebble
automaton that solves the promise problem $LOC_{n}\left( X,Y\right) $. Let $%
s\left( w\right) =\mathbf{\varepsilon }\#1^{n_{1}}\#\cdots \#1^{n_{m}}$ be
an instance of $PM_{n^{\frac{1}{4}}}$. Suppose $\mathbf{\varepsilon }\neq
0^{\ast }$. Let%
\begin{equation*}
S\left( w\right) =\mathbf{\varepsilon }^{\oplus }111X\#11\left( 10\right)
^{\ast }\$\mathbf{1}^{n_{1}}\mathbf{01\cdots 01}^{n_{m}}\mathbf{01}%
\#1111\left( 10\right) ^{\ast }\$Y.
\end{equation*}%
We construct a 2NFA $\mathcal{N}_{n}$ that, on input $s\left( w\right) $,
simulates the computation of $\mathcal{M}_{n}$ on $S\left( w\right) $. We
show that $\mathcal{N}_{n}$ separates $PM_{n^{\frac{1}{4}}}$ from $PM_{n^{%
\frac{1}{4}}}^{-}$ and that $\left\Vert \mathcal{M}_{n}\right\Vert
_{c}=\left\Vert \mathcal{N}_{n}\right\Vert _{c}$.

Let $Q$ be the set of internal of states of $\mathcal{M}_{n}$. The set of
internal states of $\mathcal{N}_{n}$ is the set: 
\begin{eqnarray*}
&&Q\times \left\{ 1,\ldots ,\left\vert X\#11\left( 10\right) ^{\ast
}\$\right\vert +\left\vert \#1111\left( 10\right) ^{\ast }\$Y\right\vert
+1\right\} \times  \\
&&\left\{ 1,\ldots ,\left\lfloor \frac{\log \left( n\right) }{4}%
\right\rfloor \right\} .
\end{eqnarray*}

\begin{itemize}
\item Suppose $\mathcal{M}_{n}$ is in state $q$ scanning the $i$-th position
of $X\#11\left( 10\right) ^{\ast }\$$. Then, $\mathcal{N}_{n}$ is in state $%
\left( q,i,\left\lfloor \frac{\log \left( n\right) }{4}\right\rfloor \right) 
$ with the input head placed on the cell $\left\vert \mathbf{\varepsilon }%
\#\right\vert +1.$

\item Suppose $\mathcal{M}_{n}$ is in state $q$ scanning the $i$-th position
of the factor $1^{n_{j}}0$. Then, $\mathcal{N}_{n}$ is in state $\left(
q,\left\vert X\#11\left( 10\right) ^{\ast }\$\right\vert +1,\left\lfloor 
\frac{\log \left( n\right) }{4}\right\rfloor \right) $ with the input head
placed on the $i$-th position of the factor $1^{n_{j}}0.$

\item Suppose $\mathcal{M}_{n}$ is in state $q$ scanning the $i$-th position
of the factor $\#1111\left( 10\right) ^{\ast }\$Y$. Then, $\mathcal{N}_{n}$
is in state 
\begin{equation*}
\left( q,\left\vert X\#11\left( 10\right) ^{\ast }\$\right\vert
+1+i,\left\lfloor \frac{\log \left( n\right) }{4}\right\rfloor \right) 
\end{equation*}
with the input head placed on the last position of $1^{n_{m}}.$

\item Let $\delta =0,1$. Suppose $\mathcal{M}_{n}$ is in state $q$ scanning
the $\left( 2i-\delta \right) $-th position of the factor $\mathbf{%
\varepsilon }^{\oplus }$. Then, $\mathcal{N}_{n}$ is in state $\left(
q,1,i\right) $ with the input head placed on the $i$-th position of $\mathbf{%
\varepsilon }.$
\end{itemize}

The set of accepting states is the set 
\begin{equation*}
\left\{ \left( q,i,j\right) :q\text{ is accepting}\right\} .
\end{equation*}%
The transition function is defined accordingly. It is easy to check that $%
\mathcal{N}_{n}$ separates $PM_{n^{\frac{1}{4}}}$ from $PM_{n^{\frac{1}{4}%
}}^{-}$.

Let $Q_{c}$ be the set of crossing states of $\mathcal{M}_{n}$. It is
important to observe that the set of crossing states of $\mathcal{N}_{n}$ is
equal to 
\begin{equation*}
Q_{c}\times \left\{ 1\right\} \times \left\{ \left\lfloor \frac{\log \left(
n\right) }{4}\right\rfloor \right\} .
\end{equation*}%
This implies the equality $\left\Vert \mathcal{M}_{n}\right\Vert
_{c}=\left\Vert \mathcal{N}_{n}\right\Vert _{c}$. We obtain:%
\begin{equation*}
2ncc\left( LOC_{n}^{+}\left( X,Y\right) ,LOC_{n}^{-}\left( X,Y\right)
\right) \geq 2ncc\left( PM_{n^{\frac{1}{4}}}^{+},PM_{n^{\frac{1}{4}%
}}^{-}\right) .
\end{equation*}%
The lemma follows.

We obtain the following corollary:
\end{proof}

\begin{corollary}
Let $X,Y\in \left\{ 0,1,\#,\$\right\} ^{\ast }$ be as above. For all $n\geq
2,$ the following inequality holds:%
\begin{equation*}
2ncc\left( LOC_{n}^{+}\left( X,Y\right) ,LOC_{n}^{-}\left( X,Y\right)
\right) \geq \left( \frac{n}{2}\right) ^{\frac{1}{8}}.
\end{equation*}
\end{corollary}

\subsection{From pattern matching to rectangle traversing}

Let $n\geq 2.$ We show that there exist two strings $X_{n},Y_{n}\in \left\{
0,1,\#,\$\right\} ^{\ast }$ such that:

\begin{enumerate}
\item $X_{n}$ avoids the pattern $\#11\left( 10\right) ^{\ast }\$$ and $%
Y_{n} $ avoids the pattern $\#1111\left( 10\right) ^{\ast }\$$.

\item If $\mathcal{M}$ is a 2NFA that separates $POS_{n}^{\ast }\left( \text{%
RA}\right) $ from $NEG_{n}^{\ast }\left( \text{RA}\right) $, then this same
automaton solves the \textit{promise problem} $LOC_{n}\left(
X_{n},Y_{n}\right) $.
\end{enumerate}

We construct $X_{n}$ and $Y_{n}$ in such a way that every positive instance
of $LOC_{n}\left( X_{n},Y_{n}\right) $, say 
\begin{equation*}
Z=\mathbf{\varepsilon }^{\left( 2\right) }111X_{n}\#11\left( 10\right)
^{\ast }\$\mathbf{1}^{n_{1}}\mathbf{0\cdots 01}^{n_{m}}\#1111\left(
10\right) ^{\ast }\$Y_{n},
\end{equation*}%
is an element of $POS_{n}^{\ast }\left( \text{RA}\right) $ and every
negative instance is an element of $NEG_{n}^{\ast }\left( \text{RA}\right) $

\begin{enumerate}
\item If we view $Z$ as an instance of $LOC_{n}\left( X_{n},Y_{n}\right) ,$
the problem we are considering is the problem of deciding whether we can
find the pattern $\mathbf{\varepsilon }$ within the set $\left\{
n_{1},\ldots ,n_{m}\right\} .$ Here, $\mathbf{\varepsilon }$ acts as a
pattern.

\item If we view $Z$ as an element of $POS_{n}^{\ast }\left( \text{RA}%
\right) \cup NEG_{n}^{\ast }\left( \text{RA}\right) $, the problem we are
considering is the problem of deciding whether the rectangle encoded by $Z$
can be crossed from left to right. Here, $\mathbf{\varepsilon }$ acts as
part of a mask.
\end{enumerate}

Then, we have to build a gadget that converts masks into patterns. Given $%
n\geq 32,$ we construct a gadget $B_{n}$ that fulfills this function.

\begin{definition}
\label{gadget}Let $n\geq 32.$ The construction of $B_{n}$ proceeds as
follows:

\begin{enumerate}
\item Given $1\leq i\leq \left\lfloor \frac{\log \left( n\right) }{4}%
\right\rfloor $, we choose a word $r_{2i-1}\in \left\{ 0,1\right\} ^{\leq n}$
encoding the rectangle $G_{i,1}$ that is defined by:

\begin{itemize}
\item $V\left( G_{i,1}\right) =\left\{ 1,\ldots ,2^{\left\lfloor \frac{\log
\left( n\right) }{4}\right\rfloor }\right\} \times \left\{ 1,2\right\} $.

\item $E\left( G_{i,1}\right) =\left\{ \left( \left( s,1\right) ,\left(
s,2\right) \right) :\text{the }i\text{-th bit of }\left\langle
s\right\rangle _{\left\lfloor \frac{\log \left( n\right) }{4}\right\rfloor }%
\text{ is equal to }1\right\} $.
\end{itemize}

\item Given $1\leq i\leq \left\lfloor \frac{\log \left( n\right) }{4}%
\right\rfloor $, we choose a word $r_{2i}\in \left\{ 0,1\right\} ^{\leq n}$
that encodes the even version of $G_{i,1}$, namely the rectangle $G_{i,0%
\text{ }}$ defined by the following dual condition:%
\begin{equation*}
E\left( G_{i,0}\right) =\left\{ \left( \left( s,1\right) ,\left( s,2\right)
\right) :\text{the }i\text{-th bit of }\left\langle s\right\rangle
_{\left\lfloor \frac{\log \left( n\right) }{4}\right\rfloor }\text{ is equal
to }0\right\} .
\end{equation*}

\item We set: 
\begin{equation*}
B_{n}=r_{1}\$\cdots \$r_{2\cdot \left\lfloor \frac{\log \left( n\right) }{4}%
\right\rfloor }\$.
\end{equation*}
\end{enumerate}
\end{definition}

\begin{definition}
Let $\mathbf{\varepsilon }\in \left\{ 0,1\right\} ^{\left\lfloor \frac{\log
\left( n\right) }{4}\right\rfloor }$. We use $i_{\mathbf{\varepsilon }}$ to
denote the integer in $\left\{ 1,...,2^{\left\lfloor \frac{\log \left(
n\right) }{4}\right\rfloor }\right\} $ that satisfies the condition $%
\left\langle i_{\mathbf{\varepsilon }}\right\rangle _{\left\lfloor \frac{%
\log \left( n\right) }{4}\right\rfloor }=\mathbf{\varepsilon }$.
\end{definition}

\begin{lemma}
\label{gadgetLemma}Let $\mathbf{\varepsilon }\in \left\{ 0,1\right\}
^{\left\lfloor \frac{\log \left( n\right) }{4}\right\rfloor }$ and let $%
B_{n}\left( \mathbf{\varepsilon }\right) $ be the rectangle encoded by the
string%
\begin{equation*}
\bigodot_{
i \leq 2\lfloor \tfrac{\log(n)}{4} \rfloor
:\,
\boldsymbol{\varepsilon}^{\oplus}[i]=1
}
r_i\$ .
\end{equation*}
The only node in the last level of $B_{n}\left( \mathbf{\varepsilon }\right) 
$ that can be reached from the first level is the node $\left( i_{%
\mathbf{\varepsilon }},\left\lfloor \frac{\log \left( n\right) }{4}%
\right\rfloor +1\right) .$
\end{lemma}

\begin{proof}
Let $\mathbf{\varepsilon }\in \left\{ 0,1\right\} ^{\left\lfloor \frac{\log
\left( n\right) }{4}\right\rfloor }$. Notice that:

\begin{enumerate}
\item The rectangle $B_{n}\left( \mathbf{\varepsilon }\right) $ is the rectangle of depth
$\lfloor \log(n)/4 \rfloor$
encoded by the string
\[
\bigodot_{i \leq \lfloor \log(n)/4 \rfloor}
G_{i,\boldsymbol{\varepsilon}[i]}\$ .
\]

\item Suppose that there exists a path from the first layer of $B_{n}\left( 
\mathbf{\varepsilon }\right) $ to the node $\left( s,\left\lfloor \frac{\log
\left( n\right) }{4}\right\rfloor +1\right) .$ This path is equal to%
\begin{equation*}
\left( s,1\right) ,\left( s,2\right) ,\ldots ,\left( s,\left\lfloor \frac{%
\log \left( n\right) }{4}\right\rfloor +1\right) .
\end{equation*}

\item For all $i\leq \left\lfloor \frac{\log \left( n\right) }{4}%
\right\rfloor $ and all $s\in \left\{ 1,\ldots ,2^{\left\lfloor \frac{\log
\left( n\right) }{4}\right\rfloor }\right\} $, the edge $\left( \left(
s,i\right) ,\left( s,i+1\right) \right) $ occurs in $B_{n}\left( \mathbf{%
\varepsilon }\right) $ if and only if 
\begin{equation*}
\mathbf{\varepsilon }\left[ i\right] =\left\langle s\right\rangle
_{\left\lfloor \frac{\log \left( n\right) }{4}\right\rfloor }\left[ i\right]
.
\end{equation*}
\end{enumerate}

The lemma follows from these three facts.
\end{proof}

\subsection{Main theorem}

We have:

\begin{theorem}
\label{2DFAs}For all $n\geq 2^{12}$, the following inequality holds:%
\begin{equation*}
2ncc\left( \text{RA,}n\right) \geq \left( \frac{n}{2}\right) ^{\frac{1}{8}}.
\end{equation*}
\end{theorem}

\begin{proof}
We show that, for all $n\geq 2^{12}$, 
\begin{equation*}
2ncc\left( POS_{n}^{\ast }\left( \text{RA}\right) ,NEG_{n}^{\ast }\left( 
\text{RA}\right) \right) \geq \left( \frac{n}{2}\right) ^{\frac{1}{8}}.
\end{equation*}

\begin{remark}
If $n\geq 2^{12},$ the inequality $\left\lfloor \frac{\log \left( n\right) }{%
4}\right\rfloor -3\geq 0$ holds.
\end{remark}

Let us suppose that $\left\lfloor \log \left( n\right) \right\rfloor $ is
divisible by $4.$

\begin{enumerate}
\item Let $\mathbf{\varepsilon }=\varepsilon _{1}\cdots \varepsilon _{\frac{%
\left\lfloor \log \left( n\right) \right\rfloor }{4}}\in \left\{ 0,1\right\}
^{\frac{\left\lfloor \log \left( n\right) \right\rfloor }{4}}$. Let $r_{%
\frac{\left\lfloor \log \left( n\right) \right\rfloor }{2}%
+2},...,r_{\left\lfloor \log \left( n\right) \right\rfloor }\in \left\{
0,1\right\} ^{\leq n}$ be words such that, for all $\frac{\left\lfloor \log
\left( n\right) \right\rfloor }{2}+1\leq 2\leq \left\lfloor \log \left(
n\right) \right\rfloor $, the string $r_{i}$ encodes a \textit{complete}
rectangle of depth $2$ and width $2^{\frac{\left\lfloor \log \left( n\right)
\right\rfloor }{4}}$ (i.e., $E\left( G_{i}\right) =\left\{ 1,\ldots ,2^{%
\frac{\left\lfloor \log \left( n\right) \right\rfloor }{4}}\right\} ^{2}$). 

\item Let $r_{1},...,r_{\frac{\left\lfloor \log \left( n\right)
\right\rfloor }{2}}$ be like in the construction of the gadget $B_{n}$, see
Definition \ref{gadget}.

\item Define:

\begin{itemize}
\item $X_{n}=1^{\frac{\left\lfloor \log \left( n\right) \right\rfloor }{4}%
-3}\#\left[ 1\right] _{\left\lfloor \log \left( n\right) \right\rfloor
}\$r_{1}\#\cdots \#\left[ \frac{\left\lfloor \log \left( n\right)
\right\rfloor }{2}\right] _{\left\lfloor \log \left( n\right) \right\rfloor
}\$r_{\frac{\left\lfloor \log \left( n\right) \right\rfloor }{2}},$ and

\item $Y_{n+1}=\left[ \frac{\left\lfloor \log \left( n\right) \right\rfloor 
}{2}+2\right] _{\left\lfloor \log \left( n\right) \right\rfloor }\$r_{\frac{%
\left\lfloor \log \left( n\right) \right\rfloor }{2}+2}\#\cdots \#\left[
\left\lfloor \log \left( n\right) \right\rfloor \right] _{\left\lfloor \log
\left( n\right) \right\rfloor }\$r_{\left\lfloor \log \left( n\right)
\right\rfloor }.$
\end{itemize}
\end{enumerate}

Notice that:

\begin{enumerate}
\item $\left[ \frac{\left\lfloor \log \left( n\right) \right\rfloor }{2}+1%
\right] _{\left\lfloor \log \left( n\right) \right\rfloor }=11\left(
10\right) ^{\frac{\left\lfloor \log \left( n\right) \right\rfloor }{2}-1}.$

\item $\left[ \frac{\left\lfloor \log \left( n\right) \right\rfloor }{2}+2%
\right] _{\left\lfloor \log \left( n\right) \right\rfloor }=1111\left(
10\right) ^{\frac{\left\lfloor \log \left( n\right) \right\rfloor }{2}-2}.$

\item Let $n_{1},\ldots ,n_{m}\leq 2^{\frac{\left\lfloor \log \left(
n\right) \right\rfloor }{4}}$ and $m\leq 2^{\log \left( n\right) }.$ Let 
\begin{equation*}
Z=\mathbf{\varepsilon }^{\oplus }111X_{n}\#11\left( 10\right) ^{\ast }\$%
\mathbf{1}^{n_{1}}\mathbf{0\cdots 01}^{n_{m}}\$\#1111\left( 10\right) ^{\ast
}\$Y_{n}.
\end{equation*}
It follows, from Proposition \ref{proposicion tecnica boba}, that $%
\#11\left( 10\right) ^{\ast }\$$ does not occur in $X_{n}$ and that $%
\#1111\left( 10\right) ^{\ast }\$$ does not occur in $Y_{n}$. Thus, we have
that the promise problem $LOC_{n}\left( X_{n},Y_{n}\right) $ is well defined
and that $Z$ is an instance of $LOC_{n}\left( X_{n},Y_{n}\right) $.

\item Let 
\begin{equation*}
Z=\mathbf{\varepsilon }^{\oplus }111X_{n}\#11\left( 10\right) ^{\ast }\$%
\mathbf{1}^{n_{1}}\mathbf{0\cdots 01}^{n_{m}}\#1111\left( 10\right) ^{\ast
}\$Y_{n}
\end{equation*}%
be an instance of $LOC_{n}\left( X_{n},Y_{n}\right) $. Suppose $\mathbf{%
\varepsilon }\neq 0^{\ast }$. If $m$ is even, we set $Z^{\ast }=Z.$ If $m$
is odd, we set:%
\begin{equation*}
Z^{\ast }=\mathbf{\varepsilon }^{\oplus }111X_{n}\#11\left( 10\right) ^{\ast
}\$\mathbf{1}^{n_{1}}\mathbf{0\cdots 01}^{n_{m}}\mathbf{01}\#1111\left(
10\right) ^{\ast }\$Y_{n}.
\end{equation*}%
We notice that $Z$ is a positive instance of $LOC_{n}\left(
X_{n},Y_{n}\right) $ if and only if $Z^{\ast }$ is a positive instance of
this promise problem. 

\item Let 
\begin{equation*}
Z=\mathbf{\varepsilon }^{\oplus }111X_{n}\#11\left( 10\right) ^{\ast }\$%
\mathbf{1}^{n_{1}}\mathbf{0\cdots 01}^{n_{m}}\#1111\left( 10\right) ^{\ast
}\$Y_{n}
\end{equation*}%
be an instance of $LOC_{n}\left( X_{n},Y_{n}\right) $. We can assume that $m$
is even. We assume this. Then, $Z$ is an element of $POS_{n}^{\ast }\left( 
\text{RA}\right) \cup NEG_{n}^{\ast }\left( \text{RA}\right) $. 

\item Let $Z$ be like in the previous item (i.e., $m$ is even). It follows
from Lemma \ref{gadgetLemma} that $Z$\ belongs to $POS_{n}^{\ast }\left( 
\text{RA}\right) $\ if and only if there exists an odd $i\leq m$ such that $%
\mathbf{\varepsilon }=\left\langle n_{i}\right\rangle _{\frac{\left\lfloor
\log \left( n\right) \right\rfloor }{4}}$. 

\item Let $Z$ be like in the previous two items. We have that $Z$ is a
positive instance of $LOC_{n}\left( X_{n},Y_{n}\right) $ if and only if $Z$
belongs to $POS_{n}^{\ast }\left( \text{RA}\right) $.
\end{enumerate}

Then, if $\mathcal{M}_{n}$ is a 2NFA that separates $POS_{n}^{\ast }\left( 
\text{RA}\right) $ from $NEG_{n}^{\ast }\left( \text{RA}\right) $ this
automaton solves the promise problem $LOC_{n}\left( X_{n},Y_{n}\right) $.
The following inequalities hold:%
\begin{equation*}
\left\Vert \mathcal{M}_{n}\right\Vert _{c}\geq \left( 2^{\frac{\left\lfloor
\log \left( n\right) \right\rfloor }{4}}\right) ^{\frac{1}{2}}\geq \left( 
\frac{n}{2}\right) ^{\frac{1}{8}}.
\end{equation*}%
The theorem is proved.
\end{proof}

\begin{corollary}
The following assertions hold:

\begin{enumerate}
\item The language RA is crossing-hard.

\item The language RA is hard-to-cross.

\item Let $\Sigma =\left\{ 0,1,\$\right\} $. For all $d<\frac{1}{8}$, all $%
\gamma >0$, and all nondeterministic pebble automaton accepting the language
RA$_{k}$, the following inequality holds asymptotically:%
\begin{equation*}
H\left( X_{\mathcal{A}_{k}}\left( \mathcal{H}_{k}^{\Sigma },n\right) \right)
\geq k\cdot d\cdot \left( 1-\gamma \right) \cdot \log \left( n\right) .
\end{equation*}

\item The sequence $\left\{ \text{RA}_{k}\right\} _{k\geq 0}$ is high in the
nondeterministic pebble hierarchy and the language RA$_{\infty }=L_{G_{0}}$
belongs to $\mathcal{NL}-\log \mathcal{CFL}.$

\item The separations $\mathcal{NL}\neq \log \mathcal{CFL}$ and $\mathcal{%
NL\neq P}$ hold.
\end{enumerate}
\end{corollary}

\textbf{Acknowledgement. }The author thanks Centro de Investigaciones
Antonio S\'{a}nchez de Cozar, vereda los Egidos, San Gil, Colombia. Most of
this work was written while the author was visiting this charming research
center in Santander mountains. Dedicated to Octavio, Marina, Carolina, the
children, and the dogs.

\end{document}